\documentclass[12pt]{article}
 \usepackage{amsfonts, amsmath, amssymb, amsthm}
 \usepackage{array, authblk}
 \usepackage[titletoc,title]{appendix}
 \usepackage{bbm, bm, booktabs}
 \usepackage{color, comment, csquotes}
\usepackage{dsfont}
\usepackage{enumerate}
 \usepackage{float, fullpage}
 \usepackage{graphics, graphicx}
 \usepackage{lscape, listings}
 \usepackage{makecell, mathrsfs, mathtools, multirow}
 \usepackage{paralist, placeins}
 \usepackage{rotating}
 \usepackage{scalerel, setspace, subdepth}
\usepackage{thmtools}
\usepackage{url}
 \usepackage{verbatim}
\usepackage{xr}
\usepackage[letterpaper, portrait, margin=1in]{geometry}
 
\usepackage{psfrag,epsf}
\usepackage[shortlabels]{enumitem}
\usepackage{epstopdf}
\usepackage{lipsum}
\usepackage{listings}
\usepackage{fancyref}
\usepackage{float}
\usepackage[font=footnotesize,labelfont=bf]{caption}

 \usepackage[round]{natbib}
 \usepackage[caption=false]{subfig}
 \usepackage[normalem]{ulem}
 \usepackage[usenames,dvipsnames]{xcolor}

 \makeatletter
 \newcommand*{\addFileDependency}[1]{
   \typeout{(#1)}
   \@addtofilelist{#1}
   \IfFileExists{#1}{}{\typeout{No file #1.}}
 }
 \makeatother
 \newcommand*{\myexternaldocument}[1]{%
     \externaldocument{#1}%
     \addFileDependency{#1.tex}%
     \addFileDependency{#1.aux}%
 }

\myexternaldocument{supplement}

 \declaretheoremstyle[
   spaceabove=6pt,
   spacebelow=6pt,
   headfont=\normalfont\itshape,
   postheadspace=1em,
   headpunct={},
   postheadspace = \newline
 ]{mystyle} 
 
 \declaretheorem[name={Remark},style=mystyle]{remark}
 \declaretheorem[name={Lemma},style=mystyle]{lemma}
 \declaretheorem[name={Theorem},style=mystyle]{theorem}
 
\declaretheorem[name={Proposition},style=mystyle]{proposition}

\def \bz{\mathbf{z}}
\def \T{\text{T}}
\def \iidsim{\stackrel{iid}{\sim}}
\newcommand*\Nzazb{N_{z_A, z_B}}

\newcommand*\Nzaplus{N_{z_A, +1_B}}
\newcommand*\Nzaminus{N_{z_A, -1_B}}
\newcommand*\Nbz{N_{\bz}}
\newcommand*\Nbzstar{N_{\bz^\star}}

\newcommand*\Nzadot{N_{z_A \boldsymbol{\cdot}}}

\newcommand*\Nzaonedot{N_{z_{A1} \boldsymbol{\cdot}}}
\newcommand*\Nzatwodot{N_{z_{A2} \boldsymbol{\cdot}}}
\newcommand*\Nplusdot{N_{+1_A \boldsymbol{\cdot}}}
\newcommand*\Nminusdot{N_{-1_A \boldsymbol{\cdot}}}

\begin{document}
\def\spacingset#1{\renewcommand{\baselinestretch}%
{#1}\small\normalsize} \spacingset{1}

\title{Causal inference from treatment-control studies having an additional factor with unknown assignment mechanism}

  \title{Causal inference from treatment-control studies having an additional factor with unknown assignment mechanism\thanks{\noindent{\textit{Email: } \texttt{np755@stat.rutgers.edu}. Kristen Hunter was supported by the Department of Defense (DoD) through the National Defense Science \& Engineering Graduate Fellowship (NDSEG) Program and
Nicole Pashley was supported by the National Science Foundation Graduate Research Fellowship under Grant No. DGE1745303 while working on this paper.
Any opinion, findings, and conclusions or recommendations expressed in this material are those of the authors and do not necessarily reflect the views of the National Science Foundation.
}}}
    \author[1]{Nicole E. Pashley}
    \affil[1]{Department of Statistics, Rutgers University}
    \author[2]{Kristen  B. Hunter}
    \affil[2]{Department of Statistics, Harvard University}
    \author[3]{Katy McKeough}
    \affil[3]{Boston Red Sox, Sports Analytics}
    \author[4]{Donald B. Rubin}
    \affil[4]{Fox School of Business, Temple University}
    \author[1]{Tirthankar Dasgupta}

\maketitle

\footnotetext{To whom correspondence should be addressed.}

\begin{abstract}
{Consider a situation with two treatments, the first of which is randomized but the second is not, and the multifactor version of this.
Interest is in treatment effects, defined using standard factorial notation.
We define estimators for the treatment effects and explore their properties when there is information about the nonrandomized treatment assignment and when there is no information on the assignment of the nonrandomized treatment.
We show when and how hidden treatments can bias estimators and inflate their sampling variances.

}
\noindent%
\end{abstract}

\spacingset{1.45} 
\section{Introduction}

Consider a randomized trial being conducted in a finite population of $N$ units to assess the causal effect of a factor (which we denote by $A$), with two levels coded as $-1_A$ and $+1_A$. The numbers of units randomly assigned to levels $-1_A$ and $+1_A$ are fixed to be $\Nminusdot$ and $\Nplusdot$ respectively.
The causal inference problem associated with such an experiment fits nicely into the framework of finite-population treatment control studies, on which there is a vast literature \citep[e.g.,][]{Imbens2015}.
However, some trials could also involve an additional causal factor (denoted by $B$), with two levels $-1_B$ and $+1_B$, one of which has to be necessarily chosen for each experimental unit, so that choice is not controlled in the design.
Th estimand of interest is naturally then the factorial effect of factor A, averaging over levels of factor $B$.
If the levels of factor $B$, like those of factor $A$, were randomly assigned to the $N$ units (with or without restrictions), the problem fits into the framework of causal inference from a randomized $2^2$ factorial experiment \citep{Dasgupta2015, Zhao2016}.
However, there may be situations in which the assignment mechanism for factor $B$ is not only unknown to the analyst, but even the actual assignments may be partially or completely missing \citep[e.g.,][Section 6.3]{rubin1991practical}.
We will consider two situations in this setting: situation I, in which the actual assignment allocation is unknown, and situation II, in which the assignment allocation is known.

An example of an experiment with the setting described above was the CORONIS study, which examined the effectiveness of 5 Cesarean section (C-section) procedural steps \citep{Coronis2013}.
Here we describe a hypothetical experiment inspired by this setting, but simplified and focusing on just two factors to facilitate understanding.
Assume the primary focus was to assess the effect of two alternatives of a specific C-section procedural step (blunt versus sharp abdominal entry) on the outcome ``maternal infectious morbidity.''
However, there was another vital procedural step with alternatives (e.g., single-layer versus double-layer closure of the uterus) that was believed to possibly affect the outcome and one of these alternatives had to be chosen for each patient at the time of the surgery.
Considering the difficulty in randomizing patients with respect to both factors (abdominal entry and closure of uterus), and to increase compliance in the experiment, hospitals were asked to randomize patients only with respect to the first factor, so the assignment mechanism with respect to the other factor was observational.
Moreover, there may only be partial information on assignment of each patient to the second factor.

To understand the key questions that arise, consider the following toy example with eight units. Four units are randomly assigned to level $-1_A$ of factor $A$ and the remaining four units to level $+1_A$.
The units are also assigned to the two levels of factor $B$ according to an \emph{unknown} probabilistic assignment mechanism.
Table \ref{tab:toy} shows two possible situations with respect to the availability of assignment information.
In both, unit 1 receives $(-1_A, -1_B)$, units 2, 3, 4 receive $(-1_A, +1_B)$, units 5, 6 receive $(+1_A, -1_B)$, and units 7, 8 receive $(+1_A, +1_B)$.
However, whereas the assignment allocation for factor $B$ is not available in situation I, it is is available in situation II.

\begin{table}[htbp]
\caption{A toy example on assignment of eight units to four treatments} \label{tab:toy}
\centering
\begin{tabular}{c|cc|cc||c|cc|cc}
\multicolumn{5}{c||}{\textbf{SITUATION I}} & \multicolumn{5}{c}{\textbf{SITUATION II}} \\ 
Unit & $-1_A$ & $+1_A$ & $-1_B$ & $+1_B$                    & Unit & $-1_A$ & $+1_A$ & $-1_B$ & $+1_B$ \\ \hline
1    &  X     &        & ? & ?        &  1   &  X     &        &   X    &        \\
2    &  X     &        & ? & ?        &  2   &  X     &        &        &   X    \\
3    &  X     &        & ? & ?       &  3   &  X     &        &        &   X    \\
4    &  X     &        & ? & ?        &  4   &  X     &        &        &   X    \\
5    &        &   X    & ? & ?        &  5   &        &   X    &   X    &        \\
6    &        &   X    & ? & ?        &  6   &        &   X    &   X    &        \\
7    &        &   X    & ? & ?        &  7   &        &   X    &        &   X    \\
8    &        &   X    & ? & ?        &  8   &        &   X    &        &   X    \\ \hline
\end{tabular}
\end{table}

We are interested in understanding the properties of estimators of the causal effect of factor $A$. Such properties can be evaluated and understood using the potential outcomes framework, introduced for randomized experiments by \citet{Neyman1923}, and extended more generally in the Rubin Causal Model \citep{rubin1974estimating} or RCM \citep{Holland1986}. Denoting the potential outcomes for unit $i (=1, \ldots, N)$ when exposed to treatments $+1_A$ and $-1_A$ by $Y_i(+1_A)$ and $Y_i(-1_A)$ respectively, the average treatment $A$ effect is
$$ \sum_{i=1}^N Y_i(+1_A) / N - \sum_{i=1}^N Y_i(-1_A) / N $$
and is often the finite population causal estimand of interest, which is our focus here.

Denote the observed outcome for unit $i$ by $Y_i^{\textnormal{\text{obs}}}$.
In situation I, a possible estimator of the causal effect of factor $A$ defined above that incorporates all available information is 
\begin{equation}
\left( \sum_{i=5}^8 Y_i^{\textnormal{\text{obs}}} \right)/4 - \left( \sum_{i=1}^4 Y_i^{\textnormal{\text{obs}}} \right)/4, \label{eq:naive_estimator}
\end{equation}
i.e., the difference between the average observed outcomes corresponding to levels $+1_A$ and $-1_A$. Classic randomization based results \citep[see, e.g.,][]{Imbens2015} give that the above estimator is an unbiased estimator of the true causal effect of factor $A$ irrespective of the assignment of units to the two levels of factor $B$ when $A$ is randomized and the potential outcomes of each unit can be expressed only as functions of the level of $A$; or in other words, when there are no ``hidden'' or ``multiple'' versions of the two treatments $-1_A$ and $+1_A$. 

The assumption that there are no hidden or multiple versions of the treatments is a part of the Stable Unit Treatment Value Assumption \citep[SUTVA,][]{Rubin1980}. The current scenario is an example where the SUTVA may be violated because a potential outcome of a patient when exposed to treatment $-1_A$ cannot be uniquely defined unless the level of factor $B$ is incorporated when factor $B$ affects the outcome. 

Several researchers have considered the problem of inferring causal effects of treatments when SUTVA is violated due to hidden or multiple versions of treatments. See \cite{Vanderweele2013} for a review. A related setting is where a mediator can be manipulated and randomized, rather than being fixed for each individual conditional on the primary treatment, effectively leading to variations in treatments \citep{rubin2004direct}. We consider a special case of a violation of SUTVA, where the ``hidden'' version is actually a pseudo treatment factor. Such a setting allows us to define potential outcomes as functions of joint interventions \citep{PearlRobins1995, Pearl2001}. However, our approach is novel in the sense that we formulate the problem in the framework of a factorial experiment \citep[see, e.g.,][]{Dasgupta2015, lu2016randomization}, define meaningful estimands accordingly, and make use of recently developed ideas in the field of causal inference from factorial experiments to evaluate properties of estimators in this setting.

Specifically, we address the following questions in a factorial design setting: (i) what is the bias of estimator (\ref{eq:naive_estimator}), (ii) under what conditions (possibly associated with the assignment mechanism of $B$) is it unbiased, (iii) what is the sampling variance of (\ref{eq:naive_estimator}), and (iv) what conditions are necessary for estimation of sampling variance. We also show how this estimator can be improved by taking the assignment data for $B$ into consideration.

In Section~\ref{sec:setup}, we introduce the notation, including potential outcomes and estimands. We describe the assumptions associated with the assignment mechanisms for the two factors under study. In Section~\ref{sec:caseI}, we discuss situation I, where we are missing the assignment information for factor $B$. The sampling distribution of the estimated factorial effect is studied and Neymanian interval estimators are proposed. In Section~\ref{sec:caseII}, we discuss situation II, in which the assignment information for factor $B$ is known. Modified point and interval estimators are proposed and compared to those for situation I.
We order the situations in this way because although it is harder to perform good inference in situation I due to the missing information, the missing information also means there is less one can do.
Therefore, in some ways the inference is simplified in situation I.
Section~\ref{sec:sim} reports numerical and simulation studies to evaluate performances of the proposed inference methods and to show how such performances are affected by various uncontrollable factors like the unknown assignment mechanism of factor $B$ and the potential outcomes matrix.
All proofs are available in the supplementary material.

\section{Notation}
\label{sec:setup}

\subsection{Potential outcomes and estimands}
\label{sec:pot}

We denote a treatment combination by $\bz = (z_A, z_B)$ where $z_{A} \in \{-1_{A}, +1_{A}\}$ and $z_{B} \in \{-1_{B}, +1_{B}\}$ represent the levels of factors $A$ and $B$ respectively. Let $\mathds{Z}$ denote the set of the four possible treatment combinations under SUTVA, or more simply, the four treatments. Let $Y_i(\bz) = Y_i(z_A, z_B)$ denote the potential outcomes for unit $i$ when exposed to treatment $\bz$, for $i=1, \ldots, N$. Each unit thus has four potential outcomes, $Y_{i}(+1_{A}, +1_{B})$, $Y_{i}(+1_{A}, -1_{B})$, $Y_{i}(-1_{A}, +1_{B})$, and $Y_{i}(-1_{A}, -1_{B})$, as shown in Table \ref{tab:pot}.

\begin{table}[ht]
    \centering \footnotesize
 \caption{Potential outcomes, unit-level and finite population-level causal estimands}
 \vspace{.1 in}
    \label{tab:pot}    
    \begin{tabular}{c c c c | c | c | c | c}
(1) & (2) & (3) & (4) & (5)= & (6)= & (7)= & (8) \\ 
    &     &     &     & (1)-(3) & (2)-(4) & $\frac{(5)+(6)}{2}$ & $\frac{(5)-(6)}{2}$  \\ \hline \hline
 \multicolumn{4}{c|}{Potential outcomes for treatment combinations} & \multicolumn{4}{c}{Causal effects} \\
 \multicolumn{4}{c|}{} & \multicolumn{2}{c|}{Conditional} & Main & Interaction \\
 $(+1_{A}, +1_{B})$ & $(+1_{A}, -1_{B})$ & $(-1_{A}, +1_{B})$ & $(-1_{A}, -1_{B})$ & $\theta_{i,A|+1_B}$ & $\theta_{i,A|-1_B}$ & $\theta_{i,A}$ & $\theta_{i,AB}$ \\ 
\multicolumn{4}{c|}{} &  &  &  &  \\ \hline
$Y_{1}(+1_{A}, +1_{B})$ & $Y_{1}(+1_{A}, -1_{B})$ & $Y_{1}(-1_{A}, +1_{B})$ & $Y_{1}(-1_{A}, -1_{B})$ & $\theta_{1,A|+1_B}$ & $\theta_{1,A|-1_B}$ & $\theta_{1,A}$ & $\theta_{1,AB}$ \\
\vdots & \vdots & \vdots & \vdots & \vdots & \vdots & \vdots & \vdots \\
$Y_{N}(+1_{A}, +1_{B})$ & $Y_{N}(+1_{A}, -1_{B})$ & $Y_{N}(-1_{A}, +1_{B})$ & $Y_{N}(-1_{A}, -1_{B})$ & $\theta_{N,A|+1_B}$ & $\theta_{N,A|-1_B}$ & $\theta_{N,A}$ & $\theta_{N,AB}$ \\ \hline \hline
& & & & & & & \\
$\bar{Y}(+1_{A}, +1_{B})$ & $\bar{Y}(+1_{A}, -1_{B})$ & $\bar{Y}(-1_{A}, +1_{B})$ & $\bar{Y}(-1_{A}, -1_{B})$ & $\theta_{A|+1_B}$ & $\theta_{A|-1_B}$ & $\theta_A$ & $\theta_{AB}$ \\
& & & & & & & \\ 
$S^2(+1_{A}, +1_{B})$ & $S^2(+1_{A}, -1_{B})$ & $S^2(-1_{A}, +1_{B})$ & $S^2(-1_{A}, -1_{B})$ & $S^2_{A|+1_B}$ & $S^2_{A|-1_B}$ & $S^2_{A}$ & $S^2_{AB}$ \\ \hline
    \end{tabular}
\end{table}

Consider the following unit-level conditional causal effects shown in columns (5) and (6) of Table \ref{tab:pot}:
\begin{equation}\label{eqn:thetaA_cond_unit}
\theta_{i,A|z_B} = Y_i(+1_A, z_B) - Y_i(-1_A, z_B), \ i=1, \ldots, N, \ z_B \in \{-1_B, +1_B\},
\end{equation}
which are the conditional effects of changing factor $A$ from $-1_A$ to $+1_A$ when the level of $B$ is held fixed at $z_B$. Averaging over the $N$ units, the average conditional effect of factor $A$ when $B$ is fixed at $z_B$ (shown in the second to last row of Table \ref{tab:pot}) is
\begin{equation}\label{eqn:thetaA_cond_overall}
\theta_{A|z_B} = \frac{1}{N} \sum_{i=1}^{N} \theta_{i,A|z_B} = \bar{Y}(+1_{A}, z_B) - \bar{Y}(-1_{A}, z_B), \ z_B \in \{+1_B, -1_B\}
\end{equation}
where $\bar{Y}(z_A, z_B) = N^{-1} \sum_{i=1}^{N} Y_i(z_A, z_B)$ is the average of the potential outcomes over all units for treatment combination $(z_A, z_B)$. Unit-level and population-level conditional effects of factor $B$ given factor $A$, denoted as $\theta_{i,B|z_A}$ and $\theta_{B|z_A}$ can be defined analogously for $i=1, \ldots, N$ and $z_A \in \{-1_A, +1_A\}$. Whereas they are not of primary interest in our problem, we shall see later that the presence or absence of such conditional effects plays an important role in the inference on the primary causal estimand of interest, which will be introduced now.

Averaging the conditional effects $\theta_{i,A|z_B}$ and $\theta_{A|z_B}$, defined in (\ref{eqn:thetaA_cond_unit}) and (\ref{eqn:thetaA_cond_overall}), over the two levels of $B$ respectively yields the unit and population level main effects of factor $A$:
\begin{eqnarray}
\theta_{i,A} = \left(\theta_{i,A|+1_B} + \theta_{i,A|-1_B} \right)/2, \ i=1, \ldots, N, \text{ and,} \nonumber \\
\theta_A = N^{-1} \sum_{i=1}^N \theta_{i,A} = \left(\theta_{A|+1_B} + \theta_{A|-1_B} \right)/2. \label{eqn:thetaA_overall}
\end{eqnarray}

The unit-level and population-level main effects, $\theta_{i,A}$ and $\theta_A$, are shown in column (7) of Table \ref{tab:pot}. Substituting (\ref{eqn:thetaA_cond_overall}) in (\ref{eqn:thetaA_overall}), we obtain
\begin{eqnarray}
\theta_{A} &=& \frac{\bar{Y}(+1_{A}, +1_{B}) + \bar{Y}(+1_{A}, -1_{B})- \bar{Y}(-1_{A}, +1_{B}) - \bar{Y}(-1_{A}, -1_{B})}{2} \label{eqn:thetaA_dot1}  \\
&=& \bar{Y}(+1_{A}, \cdot) - \bar{Y}(-1_{A}, \cdot), \label{eqn:thetaA_dot} 
\end{eqnarray}
where
\begin{align*}
\bar{Y}(z_{A}, \cdot) &= \frac{\bar{Y}(z_A,+1_{B}) + \bar{Y}(z_A, -1_{B})}{2}= \frac{1}{2N}\sum_{z_{B}}\sum_{i=1}^{N}Y_{i}(z_{A}, z_{B})
\end{align*}
represents the average of the potential outcomes for units when assigned to factor $z_{A}$ averaged over both levels of factor $B$.

Finally, we define the interaction between factor $A$ and factor $B$, at the unit or population level respectively (see the last column of Table \ref{tab:pot}), as follows:

\begin{equation}
\theta_{i, AB} = \frac{\theta_{i, A|+1_{B}} - \theta_{i, A|-1_{B}}}{2} = \frac{\theta_{i, B|+1_{A}} - \theta_{i, B|-1_{A}}}{2}, \nonumber
\end{equation}
\begin{equation}\label{eqn:thetaAB_overall}
\theta_{AB} = \frac{\theta_{A|+1_{B}} - \theta_{A|-1_{B}}}{2} = \frac{\theta_{B|+1_{A}} - \theta_{B|-1_{A}}}{2}.
\end{equation}

Note that when dealing with one treatment in a computation or result, we let $\bz = (z_{A},z_{B})$ be the general notation for the treatment combination. The explicit ($(z_{A},z_{B})$) and condensed ($\bz$) notations will be used interchangeably, based on whichever provides a clearer or cleaner result. When dealing with two different treatment combinations, $\bz$ and $\bz^*$, we will generally use the notation $\bz = (z_{A1}, z_{B1})$ and $\bz^* = (z_{A2}, z_{B2})$. 


\subsection{Assignment mechanisms and observed outcomes}
\label{sec:am}

The properties of the assignment indicator are the foundation of the randomization-based and repeated sampling properties of our treatment effect estimators.
As earlier, we assume that the level of factor $A$ is completely randomized, so that $N_{+1_A \cdot}$ units are assigned to level $+1_{A}$ of factor $A$ and the rest ($N_{-1_A \cdot}$) to level $-1_{A}$, where $N_{+1_A \cdot}$ and $N_{-1_A \cdot}$ are fixed. We assume that the assignment of the level of factor $B$ is not determined by the experimenter (i.e. is not able to be randomized by the experimenter). For instance, a doctor or the patient himself may decide the assignment level of factor $B$, possibly based on covariate values unique to the patient. For simplicity, throughout this work we assume that factor $B$ follows a Bernoulli assignment, and we assume that the probability of assignment to level $+1_B$ of factor $B$ is the same for all units and so not dependent on covariates. We denote the probability of a unit being assigned to $+1_B$ for factor $B$ as $\pi_B \in (0,1)$. We also assume that the assignment mechanisms for factor $A$ and factor $B$ are independent. As noted in literature \citep[e.g.,][Chapter 4]{Imbens2015}, a problem associated with such Bernoulli treatment assignment is the fact that the actual number of units assigned to one of the two levels of factor $B$ may be zero with a positive probability. However, this situation will not be an issue in situation I described in Section \ref{sec:caseI}, and will be addressed in situation II described in Section \ref{sec:caseII} by conditioning on the number of units assigned to each level of $B$ being strictly positive.

We define the assignment indicator $W_i(\bz)$ as a binary random variable taking value $1$ if unit $i$ is assigned to treatment $\bz$ and $0$ otherwise. The assignment indicator can be expressed as
\begin{equation}
W_{i}(\bz) = W_{i}(z_{A}, z_{B}) = W_{i,A}(z_{A})W_{i,B}(z_{B}) \nonumber
\end{equation}
for $\bz \in \{-1, +1\}^2$, $i = 1, \dots, N$, where $W_{i,A}(z_{A}) = 1$ if unit $i$ is assigned to level $z_{A}$ of factor $A$ and 0 otherwise, and $W_{i,B}(z_{B}) = 1$ if unit $i$ is assigned to level $z_{B}$ of factor $B$ and 0 otherwise. The expected value of the assignment indicators is equal to the probability of the corresponding treatment assignment. Thus, from our randomization setup, we assume that $E\left[W_{i,A}(+1_{A})\right] = N_{+1_A \cdot}/N$ and $E\left[W_{i,B}(+1_{B})\right] = \pi_{B}$, which is unknown.

Define the $N$-component vectors
$$ \mathbf{W}_A = \left(W_1(+1_A), \ldots, W_N(+1_A) \right)^{\T}, \quad \mathbf{W}_B = \left(W_1(+1_B), \ldots, W_N(+1_B) \right)^{\T},$$
where $\T$ denotes transposition.

Denoting by $\mathbf{1}$ the $N$-component vector of ones, the joint distribution of the assignment indicators can be written as

\begin{equation}
P[\mathbf{W}_A = \mathbf{w}_A, \mathbf{W}_B = \mathbf{w}_B] = \left\{   
   \begin{array}{cc}
   \left( \frac{N_{1_A+ \cdot}!N_{1_A- \cdot}!}{N!} \right) \pi_B^{\mathbf{1}^{\T} \mathbf{w}_B} (1-\pi_B)^{N - \mathbf{1}^{\T} \mathbf{w}_B}, & \mathbf{1}^{\T} \mathbf{w}_A = N_{1_A+ \cdot}, \\
   0 & \text{otherwise}.
   \end{array}
    \right. \label{eq:joint_assignment}
\end{equation}

The observed unit-level potential outcomes is
\begin{equation}
Y_{i}^{\text{obs}} = \sum_{\bz}W_{i}(\bz)Y_{i}(\bz) = \sum_{z_{A}} \sum_{z_{B}}W_{i,A}(z_{A})W_{i,B}(z_{B})Y_{i}(z_A, z_B), \nonumber
\end{equation}
and the average observed potential outcomes for a particular treatment $\bz$ is
\begin{equation}\label{eqn:ybarobs}
\bar{Y}^{\text{obs}}(\bz) = \frac{1}{\Nbz}\sum_{i=1}^{N} W_i(\bz)Y_{i}(\bz),
\end{equation}
where $\Nbz = \Nzazb$ is a random variable, potentially unobserved, representing the number of units assigned to the combination $\bz = (z_A, z_B)$. That is,
\begin{equation}\label{eqn:nbz}
\Nzazb = \sum_{i=1}^{N} W_i(z_A, z_B).
\end{equation}

We note that the joint distribution of $\{\Nbz, \bz \in \mathds{Z}\}$ is that of two independent multinomials or equivalently two Bernoullis, $Bern(N_{+1_A \cdot}, \pi_B)$ and $Bern(N_{-1_A \cdot}, \pi_B)$.


\section{Inference on the main effect of factor $A$ when assignment allocation on $B$ is unknown}\label{sec:caseI}

We now introduce an estimator for $\theta_A$, the main effect of factor $A$ in situation I, where we know that each unit receives one of two levels of factor $B$, but do not know the level actually assigned. Therefore none of the $N_{\bz}$'s defined in (\ref{eqn:nbz}) are observed and $\bar{Y}^{\text{obs}}(\bz)$ defined in (\ref{eqn:ybarobs}) cannot be computed for any $\bz$. 

However, an estimator of $\theta_A$ can be obtained by substituting estimators of $\bar{Y}(+1_A, \cdot)$ and $\bar{Y}(-1_A, \cdot)$ in (\ref{eqn:thetaA_dot}). A natural estimator of $\bar{Y}(z_A, \cdot)$, for $z_A \in\{ 1_A, +1_A\}$, is its observed average   
\begin{align*}
\bar{Y}^{\text{obs}}_{1}(z_{A}, \cdot) &= \frac{1}{\Nzadot}\sum_{i=1}^{N}W_i(z_A)Y_{i}^{\text{obs}}, \ z_A \in\{ 1_A, +1_A\},
\end{align*}
which is well-defined because $\Nzadot$ and $W_i(z_A)$ for $i=1, \ldots, N$ are known.   Consequently, we define
the following estimator of $\theta_A$:
\begin{equation}
\widehat{\theta}_{A, 1} = \bar{Y}_{1}^{\text{obs}}(+1_{A}, \cdot) - \bar{Y}_{1}^{\text{obs}}(-1_{A}, \cdot). \label{est_theta_a_est} 
\end{equation}
This estimator is in a na\"{i}ve sense the ``best we can do'' for Case I, in that it appears to incorporate all information available to us. However, ignoring the assignment of units with respect to the levels of factor B will result in $\widehat{\theta}_{A,1}$ being a generally biased estimator of $\theta_A$ without assumptions. The following result quantifies the bias and identifies conditions for unbiasedness.

\begin{proposition}[Expectation of estimator of main effect] 
\label{thm:theta1_hat_exp_bias}
The expectation of $\widehat{\theta}_{A, 1}$ is
\begin{align*}
E\left[\widehat{\theta}_{A, 1}\right]
&=  \sum_{z_{B}}E\left[W_{i,B}(z_{B})\right]\theta_{A|z_{B}},
\end{align*}
where $\theta_{A|z_{B}}$ is defined in (\ref{eqn:thetaA_cond_overall}).
\end{proposition}

\begin{remark}[Bias of the estimator and conditions for unbiasedness]
Proposition~\ref{thm:theta1_hat_exp_bias} implies that the bias of the estimator $\widehat{\theta}_{A, 1}$ is
\begin{align*}
E\left[\widehat{\theta}_{A, 1}\right] - \theta_{A}
=& \sum_{z_B}\left(E[W_{i,B}(z_B)]-\frac{1}{2}\right)\theta_{A|z_B}.
\end{align*}
The above expression implies that either of the following two conditions is sufficient for unbiasedness of $\theta_A$:
\begin{enumerate}[(a)]
\item $\pi_B = \frac{1}{2}$ for all units.
\item $\theta_{A|+1_B} = \theta_{A|-1_B}$, i.e. the conditional effects of $A$ when $B$ is held fixed at level $+1_B$ or $-1_B$ are equal. This, by (\ref{eqn:thetaAB_overall}), is equivalent to the condition that the interaction effect between factor $A$ and factor $B$, $\theta_{AB}$, is zero. A more stringent condition that implies zero interaction and thus guarantees unbiasedness of the estimator is there is no effect of factor $B$ whatsoever.
\end{enumerate}
Similar conclusions were found by \citet{delacuesta2019improving} in the context of conjoint analysis when comparing the traditional factorial estimands to so-called population effects that depend upon the distribution of other factors in the population.
\end{remark}

\noindent Next, we turn our focus on the sampling variance of our estimator. First, let the finite population variance of unit-level potential outcomes under treatment $\bz$ be 
\begin{equation} \label{eqn:s2z}
S^2(\bz) = \sum_{i=1}^{N} \frac{\left(Y_i(\bz) - \bar{Y}(\bz)\right)^2}{N-1}.
\end{equation}
Also let the finite-population covariance of unit level potential outcomes $Y_{i}(\bz)$ and $Y_{i}(\bz^*)$ be denoted
\begin{align} \label{eqn:s2zz}
S(\bz, \bz^*)
&= \sum_{i=1}^{N}\frac{\left(Y_i(\bz) - \bar{Y}(\bz)\right) \left(Y_i(\bz^\star) - \bar{Y}(\bz^\star) \right)}{N-1}.
\end{align}

The variances $S^2(\bz)$ for the four values of $\bz$ are shown in the last row and columns (1)-(4) of Table \ref{tab:pot}. Analogous to the above, we can define the variance of the unit-level conditional causal effects $\theta_{i, A|z_B}$ and that of the unconditional causal effects $\theta_{i,A}$ respectively as
\begin{eqnarray} 
S_{A|z_B}^2 &=& \sum_{i=1}^{N} \frac{\left(\theta_{i, A|z_B} - \theta_{A|z_B}\right)^2}{N-1}, \ z_B \in \{-1_B, + 1_B\}, \nonumber\\ \\
S_{A}^2 &=& \sum_{i=1}^{N} \frac{\left(\theta_{i, A} - \theta_{A}\right)^2}{N-1}. \label{eqn:s2a}
\end{eqnarray}
These variances are also shown in columns (5)-(8) of the last row of Table \ref{tab:pot}. Finally, the covariance between unit level conditional effects $\theta_{i, A|+1_B}$ and $\theta_{i, A|-1_B}$ is defined
\begin{equation} 
S_{A|+1_B, A|-1_B} = \sum_{i=1}^{N} \frac{\left(\theta_{i, A|+1_B} - \theta_{A|+1_B}\right) \left(\theta_{i, A|-1_B} - \theta_{A|-1_B}\right) }{N-1}. \nonumber 
\end{equation}
 The derivation of the sampling variance of our estimator involves some long and tedious algebraic manipulations, which are presented in Section~\ref{supp:sec_2} of the supplementary material. 

\begin{theorem}[Variance of $\widehat{\theta}_{A, 1}$]
\label{thm:theta1_hat_var}
The sampling variance of $\widehat{\theta}_{A, 1}$ is
\begin{align}\label{eqn:theta1_hat_var_alt}
Var\Big(\widehat{\theta}_{A,1}\Big)&= \sum_{z_A}\frac{\pi_B(1-\pi_B)}{\Nzadot}\left[\frac{1}{N}\sum_{i=1}^N\left(\theta_{i,B|z_A}\right)^2 +  2S^2\left((z_A, +1_B), (z_A, -1_B)\right)\right] \nonumber\\
&+\sum_{z_{A}}\frac{1}{\Nzadot}\sum_{z_B}E[W_{i,B}(z_B)]^2 S^2(z_A, z_B) - \frac{S_{A_w}^2}{N}
\end{align}
where
$S^2\left((z_A, +1_B), (z_A, -1_B)\right)$ is obtained by substituting $\bz = (z_A, +1_B)$ and $\bz^* = (z_A, -1_B)$ into (\ref{eqn:s2zz}), $S^2(z_A, z_B)$ is obtained by substituting $\bz = (z_A, z_B)$ into (\ref{eqn:s2z})  and the weighted estimator, $S_{A_w}^2$, is
\begin{eqnarray*}
S_{A_w}^2
&=& \frac{1}{N-1}\sum_{i=1}^N\left(\sum_{z_B}E[W_{i,B}(z_B)]\left(\theta_{i,A|z_B} -\theta_{A|z_B}\right)\right)^2 \nonumber \\
&=& \pi_B^2 S^2_{A|+1_B} + (1- \pi_B)^2 S^2_{A|-1_B} + 2 \pi_B (1- \pi_B) S_{A|+1_B,A|-1_B}. \label{eqn:s2X}
\end{eqnarray*}
\end{theorem}

\begin{remark}[Discussion and interpretation of Theorem \ref{thm:theta1_hat_var} for special cases]
Although the variance expression in (\ref{eqn:theta1_hat_var_alt}) looks complicated, it can be be simplified substantially under specific conditions pertaining to the potential outcomes and the assignment mechanism. First, note that when $\pi_B = 1/2$ the expression simplifies to
\begin{align}\label{eqn:theta1_hat_var_half}
Var\Big(\widehat{\theta}_{A,1}\Big)
&= \sum_{z_A}\frac{1}{2\Nzadot}\left[\frac{1}{2N}\sum_{i=1}^N\left(\theta_{i,B|z_A}\right)^2+S^2\left((z_A, +1_B), (z_A, -1_B)\right)\right] \nonumber\\
&+\sum_{z_{A}}\sum_{z_B}\frac{1}{4\Nzadot}S^2(z_A, z_B) - \frac{S^2_A}{N},
\end{align}
where $S_A^2$ is defined in (\ref{eqn:s2a}).

The second line of (\ref{eqn:theta1_hat_var_half}) is exactly that found in \cite{Dasgupta2015}, who derived the repeated sampling properties of estimated factorial effects for balanced designs where $\Nzadot = N/2$, and \cite{lu2016randomization}, who derived the repeated sampling properties of unbalanced designs. The first line of (\ref{eqn:theta1_hat_var_half}) reflects the inflation in the sample variance due to using the ``coarser'' estimator that does not take into consideration the difference between units exposed to the two levels of factor $B$.

\end{remark}
Strict additivity occurs when the treatment effects, conditional and otherwise, are the same for all units: for any set of treatments $\bz$ and $\bz^\star$, that is $Y_i(\bz) - Y_i(\bz^\star)$ is the same for each $i = 1, \dots, N$.
This assumption implies that the sample variances of potential outcomes under any treatment combination are the same. It is straightforward to see that under strict additivity, the sampling variance reduces to 
\begin{align}
\label{eqn:theta1_hat_var_strict}
Var\Big(\widehat{\theta}_{A,1}\Big)
&= \sum_{z_A}\frac{\pi_B(1-\pi_B)}{N\Nzadot}\left(\theta_{B|z_A}\right)^2+\sum_{z_{A}}\frac{S^2}{\Nzadot},
\end{align}
where $S^2=S^2(\bz)$ is the common variance of potential outcomes under each treatment.

Finally, note that, if factor $B$ has no effect whatsoever on the outcome, the sampling variance is identical to the value in a single factor experiment, because $\theta_{i,B|z_A}=0$, $\theta_{i,A|z_B} = \theta_{i,A}$ for all $z_B$, and $S^2\left((z_A, +1_B), (z_A, -1_B)\right) = S^2(z_A, z_B)$ for any $z_B$.
Thus, whereas absence of interaction between $A$ and $B$ alone is enough to guarantee unbiasedness of the estimator $\widehat{\theta}_{A,1}$, a stronger condition of null effect of $B$ is required to ensure that the estimator has the same sampling variance as under SUTVA.

\subsection{Variance estimation and asymptotic confidence intervals} \label{sec:caseIvarest}

Without information about factor $B$, we estimate sampling variance using the standard Neymanian sampling variance estimator used in typical treatment-control experiments.
First denote
\begin{align*}
s^2(z_A) = \frac{1}{N-1}\sum_{i:W_{i,A}(z_A)=1}\left(Y_i^{\text{obs}} - \bar{Y}_1^{\text{obs}}(z_{A}, \cdot)\right)^2, \ z_A \in \{-1_A, +1_A\}.
\end{align*}

Then our variance estimator is
\begin{align}\label{eqn:theta1_hat_var_hat}
\widehat{Var}\left(\widehat{\theta}_{A, 1}\right) = \sum_{z_A}\frac{s^2(z_A)}{N_{z_A\cdot}} = \frac{s^2(+1_A)}{N_{+1_A \cdot}} + \frac{s^2(-1_A)}{N_{-1_A \cdot}}.
\end{align}

The following result provides an expression for the bias of the sampling variance estimator (\ref{eqn:theta1_hat_var_hat}), which, like the bias of typical Neymanian variance estimators, is non-negative.
\begin{theorem} \label{thm:caseIvarbias}
The bias of the Neymanian sampling variance estimator given by (\ref{eqn:theta1_hat_var_hat}) is
\begin{equation}\label{eqn:theta1_hat_var_hat_bias}
E\left[\widehat{Var}(\widehat{\theta}_{A,1})\right] - Var(\widehat{\theta}_{A,1}) =\frac{1}{N(N-1)}\sum_{i=1}^N\left(\sum_{z_B}E[W_{i,B}(z_B)]\left(\theta_{i,A|z_B}-\theta_{A|z_B}\right)\right)^2. \nonumber
\end{equation}
\end{theorem}
The proof can be found in Section~\ref{supp:sec_2} of the supplementary material.
It is evident that the bias would vanish if we have
$\theta_{i,A|z_B}=\theta_{A|z_B}$ for all $i=1, \ldots, N$ and $z_B = -1_B, +1_B$, which occurs if the conditional effect of $A$ when $B$ is fixed at a particular level is the same for each unit. Note that this condition is weaker than strict additivity or no unit-level interactions.

Assuming asymptotic normality of the estimator $\widehat{\theta}_A$ \citep[see][for simple conditions for this to hold]{LiDing2017}, an approximate confidence interval for the main effect $\theta_A$ can be obtained as
$$ \widehat{\theta}_{A,1} \pm z_{\alpha/2} \sqrt{\widehat{Var}\left(\widehat{\theta}_{A, 1}\right)}, $$
where $\Phi(z_{\alpha/2})=1-\alpha/2$ and $\Phi(\cdot)$ is the cumulative distribution function (CDF) of the standard normal distribution, with $\alpha \in \{0,1\}$ defined so that the interval is a $100(1-\alpha)\%$ large sample confidence interval.


\section{Inference on the main effect of factor $A$ when assignment allocation of factor $B$ is available}\label{sec:caseII}

Next we explore an estimator for the treatment effect of factor $A$ for situation II, in which the assignment information of factor B is known.
Throughout this section we assume that $\Nbz >0$ for all $\bz$. Therefore, we always condition on $\Nbz >0$ in our exploration of situation II. We assume that the assignment mechanism for factor $A$ is independent of the assignment mechanism for factor $B$ conditional on at least one unit being assigned to each possible treatment.

Recall that in the case where the assignment of factor $B$ is unknown, $W_{i,A}(z_A)$ is independent of $W_{i,B}(z_B)$. However, if we condition on $\Nzazb$, these random variables are no longer independent. Recall from the last paragraph of Section \ref{sec:am} that the joint distribution of $\Nzazb$'s are independent Bernoulli distributed random variables. Using this fact along with the joint distribution of the assignment vector $(\mathbf{W}_A, \mathbf{W}_B)$ given by (\ref{eq:joint_assignment}), it is straightforward to see that if we condition on $\Nzazb$, we can analyze the experiment as if it were an unbalanced, completely randomized experiment, rather than independent assignment. When we remove the conditioning on $\Nzazb$, we then take into account the additional uncertainty as to the number of units assigned to each treatment group.

In this case, we can obtain an estimator of $\theta_A$ by plugging in estimators of $\bar{Y}(\bz)$ for $\bz \in \mathbb{Z}$ in the numerator of (\ref{eqn:thetaA_dot1}), because knowledge of the assignment of units to levels of $B$ makes each of these four terms estimable. A natural estimator of  $\bar{Y}(\bz)$ is  $\bar{Y}^{\text{obs}}(\bz)$ defined in (\ref{eqn:ybarobs}). Thus, we define the following estimator:

\begin{equation}
\widehat{\theta}_{A, 2}
= \frac{\bar{Y}^{\text{obs}}(+1_{A}, +1_{B})- \bar{Y}^{\text{obs}}(-1_{A}, +1_{B}) + \bar{Y}^{\text{obs}}(+1_{A}, -1_{B})  - \bar{Y}^{\text{obs}}(-1_{A}, -1_{B})}{2}. \nonumber 
\end{equation}

The sampling properties of this estimator are summarized in the following two results.

\begin{proposition}[Expectation of $\widehat{\theta}_{A, 2}$]
\label{thm:theta2_hat_exp_bias}
The estimator $\widehat{\theta}_{A, 2}$ is an unbiased estimator of $\theta_A$.
\end{proposition}

\medskip

\begin{theorem}[Variance of $\widehat{\theta}_{A, 2}$]
\label{thm:theta2_hat_var}
The variance of $\widehat{\theta}_{A, 2}$ is
\begin{align}\label{eqn:theta2_hat_var_rewrite}
Var\left(\widehat{\theta}_{A, 2}\Big|N_{\bz}>0 \text{ }\forall \bz\right)&=
\sum_{\bz} E\left[\frac{1}{4N_\bz}\Big|N_{\bz}>0\right]S^2(\bz) -\frac{1}{N}S^2_A.
\end{align}
\end{theorem}
For proof of this result, see Section~\ref{supp:caseII} of the supplementary material.

\begin{remark}[Discussion on the sampling properties of $\widehat{\theta}_{A,2}$]

Proposition~\ref{thm:theta2_hat_exp_bias} implies that when the actual assignment of units to levels of factor $B$ is known, we can unbiasedly estimate $\theta_{A}$, even without knowing the exact assignment mechanism of factor $B$. This result illustrates how important it is to record all relevant information when running an experiment. The sampling variance result is equivalent to that found by 
\citet{lu2016randomization} for unbalanced $2^K$ factorial designs (which is the design being used conditional on $N_{\bz}$) and generalizes the result found in \cite{Dasgupta2015} for balanced $2^2$ factorial experiments to the case of random $N_{\bz}$'s resulting in possibly unbalanced designs.

\end{remark}

As in situation I, we explore the properties of the sampling variance of the estimator $\widehat{\theta}_{A,2}$ under strict additivity. Looking back at (\ref{eqn:s2a}), we have $S^2_A = 0$ as all $\theta_{i, A}$ are the same ($\theta_{i, A} = \theta_A$, $i=1, \ldots, N$). Furthermore, under strict additivity, all of the $S^2(\bz)$ are the same, so we can denote the common finite-population variance of potential outcomes under any treatment $\bz$ as $S^2$. Hence, under strict additivity, (\ref{eqn:theta2_hat_var_rewrite}) reduces to
\begin{equation}
Var\left(\widehat{\theta}_{A, 2}\Big|N_{\bz}>0 \text{ }\forall \bz\right)=
\sum_{\bz} E\left[\frac{1}{4N_\bz}\Big|N_{\bz}>0\right]S^2. \nonumber
\end{equation}

\subsection{Estimation of sampling variance and interval estimation of main effect of $A$}

Next, we consider the problem of estimating the quantity $Var\left(\widehat{\theta}_{A, 2}\right)$. 
From the expression of variance derived in (\ref{eqn:theta2_hat_var_rewrite}), and noting that $S^2_A$ is not estimable without assumptions like strict additivity, we can create a conservative estimator by plugging in estimators of $S^2(\bz)$ and $E\left(N_{\bz}^{-1}| N_{\bz} >0 \right)$ in (\ref{eqn:theta2_hat_var_rewrite}). Instead of estimating $E\left(N_{\bz}^{-1}| N_{\bz} >0 \right)$, the most straightforward estimator for the sampling variance of $\widehat{\theta}_{A, 2}$ would be the one obtained by conditioning on the observed number of units assigned to each treatment level, $\Nbz$. Thus, we propose the variance estimator, also used by \citet{lu2016randomization},
\begin{equation}
\widehat{Var}\left(\widehat{\theta}_{A, 2}\right)=\sum_{\bz} \frac{1}{4\Nbz}s^2(\bz), \label{eqn:var_edt_caseII}
\end{equation}
where
$$ s^2(\bz) = \frac{1}{\Nbz - 1}\sum_{i:W_i(\bz) = 1}\left[Y_i^{\text{obs}} - \bar{Y}^{\text{obs}}(\bz)\right]^2$$
is the sample variance of the observed outcomes for treatment combination $\bz$, and is an unbiased estimator of $S^2(\bz)$ conditional on $N_{\bz} > 0$ for all $\bz$.
For this estimator to actually be defined, we require $N_{\bz} \geq 2$ for all $\bz$.

This estimation procedure is analogous to that in an experimental setup where the design is a Bernoulli experiment with a single factor, but one analyzes the experiment as if it were completely randomized by conditioning on the number of treated units. If we have a method that creates valid confidence intervals (i.e. with correct frequentist coverage) conditional on the number of treated units, then we will get (unconditional) valid confidence intervals over the original assignment mechanism using this method \citep{pashley2020conditional}.
See also \cite{branson2019randomization} and \cite{hennessy2016conditional} for exploration of related conditioning ideas under Fisherian inference.

In Section~\ref{supp:alt_var_est} of the supplementary material, we explore an alternative variance estimation strategy which aims to estimate the overall variance of $\widehat{\theta}_{A, 2}$, not conditional on $\Nbz$ itself, though still conditioning on $\Nbz>0$.

\section{Simulation Results}
\label{sec:sim}

We compare the performance of the estimators for situation I and situation II in numerical and simulation studies.
We compare three potential outcome models: (1) strict additivity with $\theta_{AB} = 0$, (2) strict additivity with $\theta_{AB} \neq 0$, and (3) moderately, positively correlated potential outcomes, such that each pair of potential outcomes has a fixed positive correlation.

When assuming strict additivity, we generate $\epsilon_i \iidsim N(0,\sigma^2)$, and then set the potential outcomes as
\begin{align*}
Y_i(-1_A,-1_B) &= \epsilon_i + \theta_{AB}^*\\
Y_i(+1_A,-1_B) &= \epsilon_i + \theta_A^*\\
Y_i(-1_A,+1_B) &= \epsilon_i + \theta_B^*\\
Y_i(+1_A, +1_B) &= \epsilon_i + \theta_A^* + \theta_B^* + \theta_{AB}^*,
\end{align*}
where the $\theta_A^*, \theta_B^*$, and $\theta_{AB}^*$ are hypothetical population parameters, as opposed to $\theta_A, \theta_B$, and $\theta_{AB}$ which are realized finite population parameters based on one draw from the population.

When assuming positive correlation, we draw the potential outcomes according to the following model:
\begin{align*}
\begin{pmatrix}
Y_i(-1_A,-1_A) 	\\ Y_i(+1_A, -1_B) 	\\ Y_i(-1_A,+1_B) 	\\ Y_i(+1_A,+1_B) \\
\end{pmatrix} \sim N_4\left(
\begin{pmatrix}
\theta_{AB}^* 	\\ \theta_A^* 		\\ \theta_B^* 			\\ \theta_A^* + \theta_B^* + \theta_{AB}^* \\
\end{pmatrix},
\begin{pmatrix}
\sigma^2 		& \rho \sigma^2 	& \rho \sigma^2 		& \rho \sigma^2 \\
\rho \sigma^2 	& \sigma^2 			& \rho \sigma^2 		& \rho \sigma^2 \\
\rho \sigma^2 	& \rho \sigma^2 	& \sigma^2 			& \rho \sigma^2 \\
\rho \sigma^2 	& \rho \sigma^2 	& \rho \sigma^2 		& \sigma^2 \\
\end{pmatrix}
\right)
\end{align*}

We assume $N = 100$, $N_{+1_A \cdot} = N_{-1_A \cdot} = 50$, $\sigma^2 = 1$, and $\rho = 0.4$.
The simulation is structured as follows:

\begin{enumerate}
\item Fix a potential outcome model, and draw the potential outcomes $\mathbf{Y}$. 
\item Calculate the true finite sample factorial effect of factor A, $\theta_A$.
\item For each value $\pi_B \in (0.05, 0.1, \dots, 0.9, 0.95)$:
	\begin{enumerate}
		\item Generate $1,000$ different assignment vectors $\mathbf{W}$.
		For situation II, any assignment vectors with any $\Nbz \leq 2$ were rejected so that all estimators are well-defined.
		\item Calculate point estimates $\widehat{\theta}_A$, expected values $E\left[\widehat{\theta}_{A}\right]$, and sampling variance estimate $\widehat{Var}\left(\widehat{\theta}_{A}\right)$ for each $\mathbf{W}$.
		For situation I, we use the sampling variance estimator $\widehat{Var}\left(\widehat{\theta}_{A, 1}\right)$ (Equation~\ref{eqn:theta1_hat_var_hat}). 
		For situation II, we use the sampling variance estimator $\widehat{Var}\left(\widehat{\theta}_{A,2}\right)$ (Equation~\ref{eqn:var_edt_caseII}).
	\end{enumerate}
\item Evaluate performance.
For each of the following quantities, the mean over the $1,000$ assignment vectors is reported for each value of $\pi_B$:
	\begin{enumerate}
		\item Coverage: $\mathds{1}\left(\theta_A \in \widehat{\theta}_A \pm 1.96\sqrt{\widehat{Var}\left(\widehat{\theta}_{A}\right)}\right)$, where $\mathds{1}(x)$ is an indicator variable for the condition $x$.
		\item Interval width: $2 \times 1.96\sqrt{\widehat{Var}\left(\widehat{\theta}_{A}\right)}$.
	\end{enumerate}
\end{enumerate}

We note two features of the simulations.
First, the randomness in the simulations comes solely from the assignment vector; the potential outcomes are only generated three times, one for each potential outcome model.
Second, when calculating $E\left[\widehat{\theta}_{A}\right]$ and $Var\left(\widehat{\theta}_{A}\right)$, we use the numerical equations outlined in this paper, rather than empirical estimates of expectation and variance.

In a model with strict additivity and no interaction effect, the situation I and situation II estimators of $\theta_A$ show similar performances (Figure \ref{fig:sim}).
The data generating values are $\theta_{AB}^* = 0$, $\theta_A^* = 2$, and $\theta_B^* = 2$, although the finite-population values differ slightly from these data-generating values due to the errors $\epsilon_i$ generated.
The situation I and situation II confidence interval methods have similar coverage.
However, the interval width using the situation I estimators is smaller at more extreme values of $\pi_B$ (close to $0$ or $1$).
If most units have the same factor $B$ assignment, then $\bar{Y}(z_A,\cdot)$ values will vary less.
In contrast, at extreme values of $\pi_B$, the situation II estimator becomes unstable because very small groups of units are being used in the sampling variance estimator.
Thus the coverage dips below 95\% at extreme values of $\pi_B$ for situation II because of these unstable sampling variance estimates.
Because situation I estimators never rely on information from factor $B$, and thus they pool these small groups into larger groups, they do not suffer from this instability.
When we assume a nonzero interaction effect, the situation I estimator performs poorly in comparison to the situation II estimator.
The data generating values are $\theta_{AB}^* = 2$, $\theta_A^* = 2$, and $\theta_B^* = 2$.
The coverage for situation I is very low except when $\pi_B =  0.5$, consistent with the result that the estimator is only unbiased under this condition.
Additionally, the situation I estimator has larger average interval widths than in the first simulation, whereas the situation II estimator interval widths remain the same.
The larger interval widths could be due to the interaction term resulting in more variability in the potential outcomes under different treatment combinations.
The situation II estimator is unaffected by the nonzero interaction term because the variances for the potential outcomes are calculated separately for each treatment combination, while in situation I treatment combinations with the same factor $B$ assignment are grouped together. 

Relaxing strict additivity results in similar performance to situations with strict additivity and a nonzero interaction.
In the third model, rather than strict additivity, we assume the potential outcomes are positively correlated, with any pair of outcomes $Y_i(\bz)$ and $Y_i(\bz^\star)$ having a correlation of $\rho = 0.4$, and the same data-generating values $\theta_{AB}^* = 2$, $\theta_A^* = 2$, and $\theta_B^* = 2$.
This simulation setup slightly breaks the assumption of strict additivity, although by having positive correlation between each pair of outcomes, we only introduce a small amount of variance among the differences of the unit-level potential outcomes $Y_i(\bz) - Y_i(\bz^\star)$.
Without strict additivity, our sampling variance estimators are biased and overestimate the true sampling variance, which is shown by the slight overcoverage.
However, the simulation results without strict additivity are generally very similar to the simulation results with strict additivity and an interaction, so we find that slightly breaking the assumption of strict additivity has only a minor impact on empirical performance of the estimators.

For additional simulation results, including mean squared error (MSE), average relative bias, and average relative bias of the variance estimator, see Section~\ref{sec:sim_add} of the supplementary material.

\section{Discussion}
We have explored the setting in which a researcher has two treatments of interest (factor $A$ and factor $B$), but is in control of the assignment mechanism of only one factor (factor $A$).
We considered two situations: one in which the experimenter does not know the assignment of factor $B$ (a violation of SUTVA) and one in which the experimenter does know the assignment of factor $B$.
In the situation where the assignment of factor $B$ is unknown, the estimator for the main effect of factor $A$ is the simple difference in means, which is only unbiased if the assignment of factor $B$ is balanced (i.e., each unit has equal probability of being assigned to either level of factor $B$) or there is no interaction between factor $A$ and factor $B$.
In the situation where the assignment of factor $B$ is known, the usual factorial estimator for the main effect of factor $A$ is unbiased even in the aforementioned contexts.
We further proposed sampling variance estimators for both estimators.

We conducted a numerical simulation to illustrate the properties of these estimators.
In the situation with unknown factor $B$ assignment, a balanced factor $B$ assignment mechanism results in zero bias, but the highest sampling variance because $\bar{Y}(z_A, \cdot)$ is an even mixture of those assigned to levels $+1_B$ and $-1_B$ of factor $B$.
In the situation with known factor $B$ assignment, the opposite holds true for the sampling variance; the sampling variance in that situation is minimized with a balanced factor $B$ assignment mechanism.

The work in this paper points to many avenues of future exploration.
For example, we could extend this work to a context in which the probability of being assigned to factor $B$ is not constant, but instead depends on characteristics of the units.
Such an assignment could be more realistic.
In the situation where assignment to factor $B$ is observed, we could then use methods such as propensity score models to build better estimates of treatment effects.
Additionally, it is of interest to explore a setting in which we have $K$ treatments and multiple blocks, with each block randomizing a different subset of the $K$ treatments, which more closely resembles the setting of \citet{Coronis2013}.
In this setup, a Bayesian hierarchical model might be useful to share information across sites.

\section{Software}
\label{sec:software}

We include simulation code in a GitHub repository at \url{https://github.com/kristenbhunter/hidden_treat}.

\bibliographystyle{apalike}
\bibliography{julius}

\begin{figure}[!p]
\centering
  \includegraphics[width=6.5in]{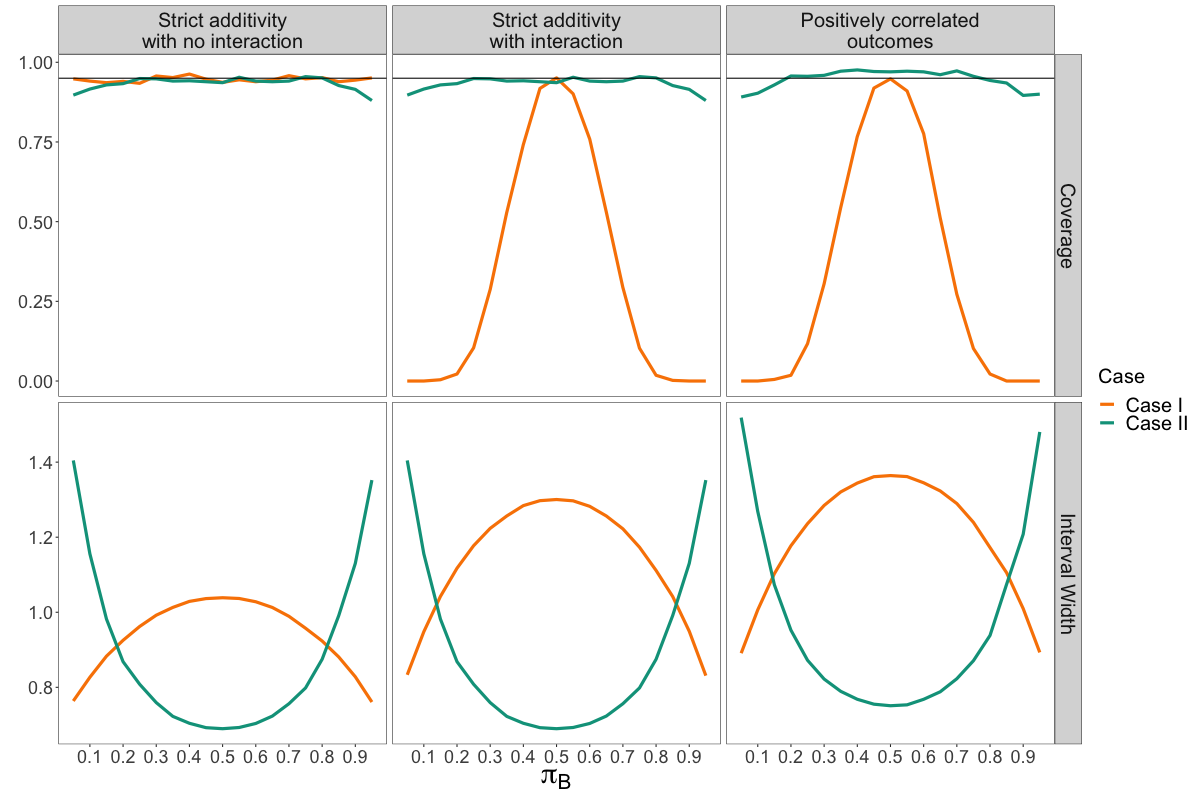}
  \caption{\label{fig:sim}Comparison of performances of situation I and situation II estimators}
\end{figure}

\begin{appendices}
\renewcommand\appendixname{Supplementary Material}
\section{Alternative sampling variance estimator}\label{supp:alt_var_est}

Instead of using variance estimator (\ref{eqn:var_edt_caseII}), in which $E[N_{\bz}^{-1}|N_{\bz}>0]$ is estimated by $1/\Nbz^{\text{obs}}$, we can derive an exact expression for $E[N_{\bz}^{-1}|N_{\bz}>0]$ and obtain its estimator. Note that $N_{+1_A \cdot}$ and $N_{-1_A \cdot}$ are \emph{not} random variables because they are fixed by the experimenter. However, conditional on $N_{\bz} >0$ for all $\bz$, $N_{z_A,z_B}$ is a truncated binomial random variable with parameters $\left(N_{z_A\cdot},\pi_B \right)$ taking values $1, \ldots, (N_{z_A\cdot}~-1)$. The following result provides an expression for $E[N_{\bz}^{-1}|N_{\bz}>0]$:

\begin{lemma} \label{lem:mean_of_Nz_inv}
For $z_A \in \{-1_A, +1_A\}$ and $z_B \in \{-1_B, +1_B\}$,
\begin{align}
    E\left[\frac{1}{N_{z_A, z_B}}|N_{\bz}>0\right] &= \frac{1}{1-\pi_B^{\Nzadot} - (1-\pi_B)^{\Nzadot}}\sum_{n=1}^{\Nzadot-1}\frac{1}{n}{\Nzadot \choose n}E[W_{i,B}(z_B)]^n(1-E[W_{i,B}(z_B)])^{\Nzadot - n} \label{eqn:mean_of_Nz_inv}
\end{align}
\end{lemma}

A reasonable way to estimate $\pi_B$ is to pool the information from the units where $z_A = +1_A$ and $z_A = -1_A$, to get $\widehat{\pi}_B = \frac{n_{\cdot +1_B}}{N}$, where $n_{\cdot +1_B}$ is the observed number of units assigned to level $+1_B$ of $B$.  We can then substitute the estimator $\widehat{\pi}_B$ in (\ref{eqn:mean_of_Nz_inv}) and manually solve the sum to get an estimator of the expected value of $1/\Nbz$, and consequently the following plug-in estimator for the sampling variance of the estimator $\widehat{\theta}_{A, 2}$,
\begin{equation}\label{eqn:var_2_caseII}
\widehat{Var}_2\left(\widehat{\theta}_{A, 2}\right) = \left[ \sum_{\bz} \frac{\widehat{E}[1/N_{\bz}]}{4}  s^2(\bz)\right].
\end{equation}

This estimator is biased.
This bias is the price we pay for not knowing $\pi_B$ and having to estimate $E[N_{\bz}^{-1} | N_{\bz} >0]$, a non-linear function of $\pi_B$.
 However, as we shall see in our simulations, the bias does not have severe adverse effects on the coverage of the asymptotic confidence intervals for $\theta_A$ that can be generated as described in Section \ref{sec:caseIvarest}, even for moderately small population sizes.
 
\newpage
\section{Additional simulation results}
\label{sec:sim_add}
 
\begin{figure}[h!]
\centering
  \includegraphics[width=5in]{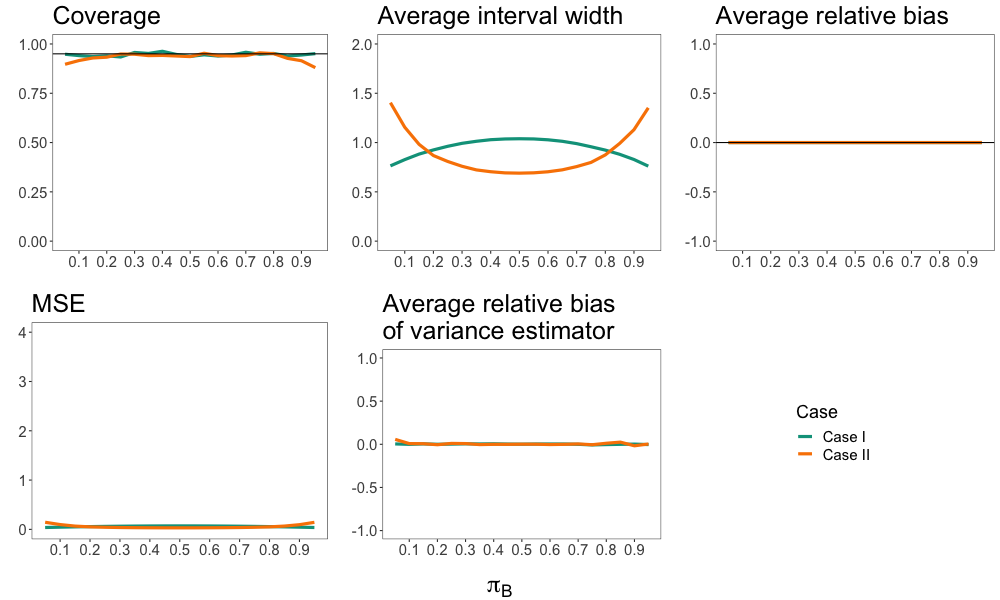}
  \caption{Strict additivity with no interaction}
\end{figure}

\begin{figure}[h!]
\centering
  \includegraphics[width=5in]{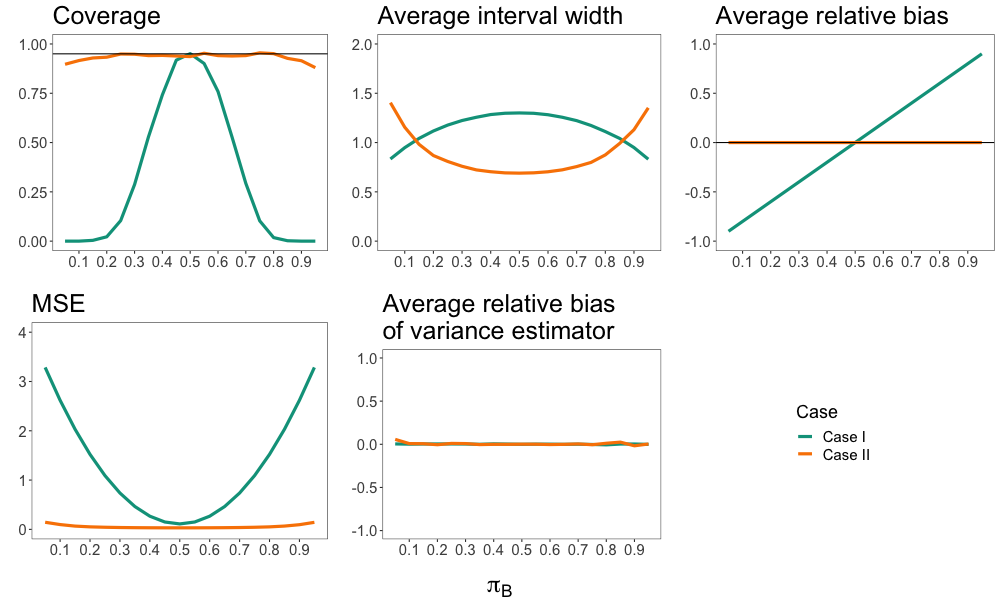}
  \caption{Strict additivity with interaction}
\end{figure}

\begin{figure}[h!]
\centering
  \includegraphics[width=5in]{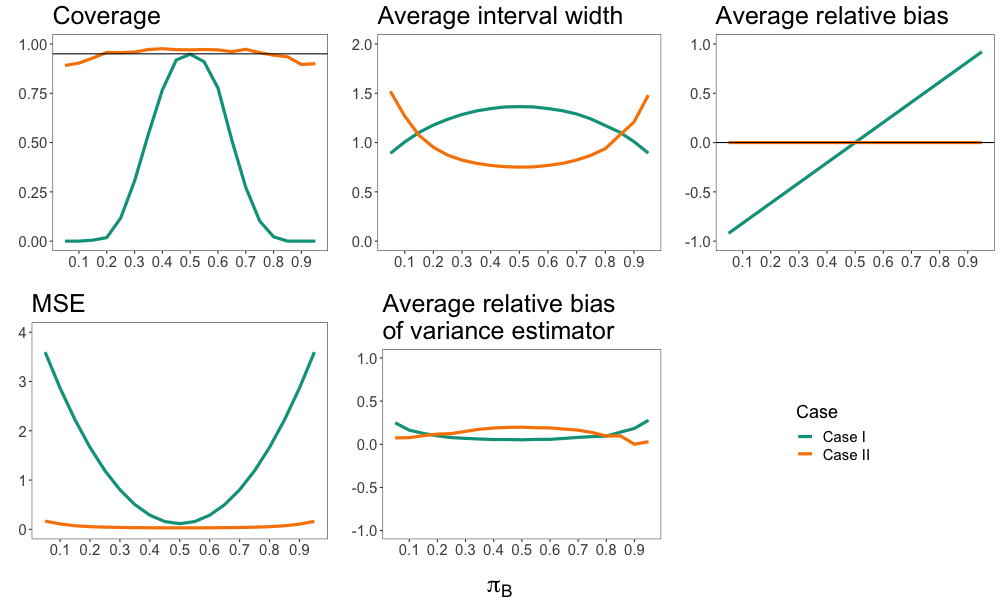}
  \caption{Positively correlated outcomes}
\end{figure}

\section{Properties of the assignment indicator}\label{supp:prop_assign}
The second order moments of the assignment indicator variables are summarized in the following two lemmas.

\begin{lemma}
\begin{align*}
Cov\left(W_{i,A}(z_{A1}), W_{i',A}(z_{A2})\right) =\frac{N_{+1_A \cdot} N_{-1_A \cdot}} {N^2}(-1)^{\mathds{1}(z_{A1}\neq z_{A2})}\left(\frac{-1}{N-1}\right)^{\mathds{1}(i \neq i')}.
\end{align*}
\end{lemma}

\begin{proof}
\begin{align*}
Cov\Big(W_{i,A}(z_{A1}), W_{i',A}(z_{A2})\Big) &= E\Big[W_{i, A}(z_{A1})W_{i', A}(z_{A2})\Big] -E\Big[W_{i, A}(z_{A1})\Big]E\Big[W_{i', A}(z_{A2})\Big]\\
&= E\Big[W_{i, A}(z_{A1})W_{i', A}(z_{A2})\Big] - \frac{\Nzaonedot\Nzatwodot}{N^2}
\end{align*}

\begin{itemize}
\item If $i = i'$:
\begin{itemize}
\item If $z_{A1} = z_{A2}$,
\begin{align*}
Cov\Big(W_{i,A}(z_{A1}), W_{i,A}(z_{A1})\Big) &= E\Big[W_{i, A}(z_{A1})^2\Big] - \frac{\Nzaonedot^2}{N^2} \\
&= \frac{\Nplusdot\Nminusdot}{N^2}.
\end{align*}
\item If $z_{A1} \neq z_{A2}$,
\begin{align*}
Cov\Big(W_{i,A}(z_{A1}), W_{i,A}(z_{A2})\Big) &=E\Big[W_{i,A}(z_{A1})W_{i,A}(z_{A2})\Big] - \frac{\Nzaonedot\Nzatwodot}{N^2}\\
&= - \frac{\Nplusdot\Nminusdot}{N^2}.
\end{align*}
\end{itemize}
\item If $i \neq i'$:
\begin{itemize}
\item If $z_{A1} = z_{A2}$,
\begin{align*}
E\Big[W_{i, A}(z_{A1})W_{i', A}(z_{A1})\Big] &=  P\Big(W_{i',A}(z_{A1}) = 1\Big|W_{i,A}(z_{A1}) = 1\Big)E\Big[W_{i,A}(z_{A1})\Big]\\
&= \frac{\Nzaonedot(\Nzaonedot - 1)}{N(N-1)}\\
Cov(W_{i,A}(z_{A}), W_{i',A}(z_{A})) &=  \frac{\Nzaonedot(\Nzaonedot - 1)}{N(N-1)} - \frac{\Nzaonedot^2}{N^2}\\
&=  -\frac{\Nplusdot\Nminusdot}{N^2(N-1)}.
\end{align*}
\item If $z_{A1} \neq z_{A2}$:
\begin{align*}
E\Big[W_{i,A}(z_{A1})W_{i',A}(z_{A2})\Big] &= P\Big(W_{i',A}(z_{A2}) = 1\Big|W_{i,A}(z_{A1}) = 1\Big)P\Big(W_{i,A}(z_{A1})=1\Big)\\
&= \frac{\Nplusdot\Nminusdot}{N(N-1)}\\
Cov\Big(W_{i,A}(z_{A1}), W_{i',A}(z_{A2})\Big) &= \frac{\Nzaonedot\Nzatwodot}{N(N-1)} - \frac{\Nzaonedot\Nzatwodot}{N^2}\\
&= \frac{\Nplusdot\Nminusdot}{N^2(N-1)}.
\end{align*}
\end{itemize}
\end{itemize}

Putting this all together,
\[Cov\left(W_{i,A}(z_{A1}), W_{i',A}(z_{A2})\right) =\frac{\Nplusdot\Nminusdot}{N^2}(-1)^{\mathds{1}(z_{A1}\neq z_{A2})}\left(\frac{-1}{N-1}\right)^{\mathds{1}(i \neq i')}.\]
\end{proof}


\begin{lemma}[Properties of the joint assignment indicator]
\label{lem:w_prop}
\begin{align*}
E\left[W_{i}(z_A,z_B)\right] &= \left(\frac{\Nzadot}{N}\right)E\left[W_{i,B}(z_B)\right]\\
\label{eqn:w_var}Var\left(W_{i}(z_{A}, z_{B})\right) &= \frac{E\left[W_{i,B}(z_B)\right]\Nzadot\left(N-E\left[W_{i,B}(z_B)\right]\Nzadot\right)}{N^2}
\end{align*}

\begin{align*}
&Cov\left(W_i(z_{A1}, z_{B1}), W_{i'}(z_{A2}, z_{B2})\right)\\
\nonumber &= \begin{dcases}
\frac{E\left[W_{i,B}(z_{B1})\right]E\left[W_{i',B}(z_{B2})\right]\Nplusdot\Nminusdot}{N^2(N-1)}
	& \text{if } i \neq i' \text{ and } z_{A1} \neq z_{A2}\\
\frac{-E\left[W_{i,B}(z_{B1})\right]E\left[W_{i',B}(z_{B2})\right]\Nplusdot\Nminusdot}{N^2(N-1)}
	& \text{if } i \neq i' \text{ and } z_{A1} = z_{A2} \\
\frac{-E\left[W_{i,B}(z_{B1})\right]E\left[W_{i,B}(z_{B2})\right]\Nplusdot\Nminusdot}{N^2\vphantom{(N-1)}}
	& \text{if } i = i' \text{ and } \bz \neq \bz^* \\
\frac{E\left[W_{i,B}(z_{B})\right]\Nzadot\left(N-E\left[W_{i,B}(z_{B})\right]\Nzadot\right)}{N^2\vphantom{(N-1)}}
	& \text{if }  i = i' \text{ and }  \bz = \bz^*.
\end{dcases}
\end{align*}

\end{lemma}

\begin{proof}
Using the properties of a Bernoulli distribution,
\begin{align*}
Var\Big(W_i(z_A, z_B)\Big) =& E\left[W_{i}(z_A, z_B)\right]\left(1-E\left[W_{i}(z_A,z_B)\right]\right)\\
&= \frac{E\left[W_{i,B}(z_B)\right]\Nzadot\left(N-E\left[W_{i,B}(z_B)\right]\Nzadot\right)}{N^2}.
\end{align*}

Recall from the previous section that
\[Cov\left(W_{i,A}(z_{A1}), W_{i',A}(z_{A2})\right) =\frac{\Nplusdot\Nminusdot}{N^2}(-1)^{\mathds{1}(z_{A1}\neq z_{A2})}\left(\frac{-1}{N-1}\right)^{\mathds{1}(i \neq i')}.\]

Then we can again break the covariance terms into cases by $i$ and $\bz$.

For $i \neq i'$,
\begin{align*}
Cov\Big(W_i(z_{A1}, z_{B1}), W_{i'}(z_{A2}, z_{B2})\Big) =& E\Big[W_i(z_{A1}, z_{B1})W_{i'}(z_{A2}, z_{B2})\Big]-E\Big[W_i(z_{A1}, z_{B1})\Big]E\Big[W_{i'}(z_{A2}, z_{B2})\Big]\\
=& E\Big[W_{i, B}(z_{B1})\Big]E\Big[W_{i', B}(z_{B2})\Big]Cov\Big(W_{i, A}(z_{A1}),W_{i', A}(z_{A2})\Big).
\end{align*}


For $i = i'$ if $z_{A1} \neq z_{A2}$ or $z_{B1} \neq z_{B2}$ then
\begin{align*}
Cov\Big(W_i(z_{A1}, z_{B1}), W_{i}(z_{A2}, z_{B2})\Big) &= -E\Big[W_i(z_{A1}, z_{B1})\Big]E\Big[W_{i}(z_{A2}, z_{B2})\Big]\\
&=-\frac{\Nzaonedot\Nzatwodot}{N^2}E\Big[W_{i,B}(z_{B1})\Big]E\Big[W_{i, B}(z_{B2})\Big].
\end{align*}

Otherwise, for $i = i'$ if $z_{A1} = z_{A2}$ and $z_{B1} = z_{B2}$,

\begin{align*}
Cov\Big(W_i(z_{A1}, z_{B1}), W_{i}(z_{A2}, z_{B2})\Big) = Var\Big(W_i(z_A, z_B)\Big).
\end{align*}

\end{proof}

\section{Proof of results in Section~\ref{sec:caseI}: Case II estimators}\label{supp:sec_2}


\subsection*{Proof of Theorem~\ref{thm:theta1_hat_var}}
 \label{supp:theta1_hat_var}
 
 Before the proof, first note the following two general properties for any potential outcomes $U_i$ (with $\overline{U}$ indicating the relevant finite-population mean) that are defined based on assignment to factor $A$ alone (similar to a single factor experiment) or based on assignment to both $A$ and $B$:
\begin{align*}
    Var\left(\frac{1}{\Nzaonedot}\sum_{i=1}^n W_{i,A}(z_A)U_i(z_A)\right) &= \frac{N-\Nzaonedot}{N\Nzaonedot}\frac{1}{N-1}\sum_{i=1}^n\left(U_i(z_A) - \overline{U}(z_A)\right)^2\\
    Var\left(\sum_{z_B}W_{i,B}(z_B)U_i(z_A, z_B)\Bigg|z_A\right) &= \sum_{z_B}\pi_B(1-\pi_B)U_i(z_A, z_B)^2 - 2\pi_B(1-\pi_B)U_i(z_A, +1_B)U_i(z_A, -1_B)\\
    &=\pi_B(1-\pi_B)\left(U_i(z_A, +1_B)-U_i(z_A, -1_B)\right)^2
\end{align*}

\begin{proof}
Let $\bm{W}_A(z_A)$ be a vector of indicators for whether each unit is assigned to level $z_A$ of factor $A$.
\begin{align*}
&Var\left(\overline{Y}^{\text{obs}}_{1}(z_{A}, \cdot)\right)\\
 &= E\left[Var\left(\overline{Y}^{\text{obs}}_{1}(z_{A}, \cdot)\Big|\bm{W}_A(z_A)\right)\right] + Var\left(E\left[\overline{Y}^{\text{obs}}_{1}(z_{A}, \cdot)\Big|\bm{W}_A(z_A)\right]\right)\\
&= E\left[\frac{1}{\Nzadot^2}Var\left(\sum_{i=1}^NW_{i,A}(z_{A})\sum_{z_{B}}W_{i,B}(z_{B})Y_i(z_{A},z_{B})\Bigg|\bm{W}_A(z_A)\right)\right]\\
& + Var\left(E\left[\frac{1}{\Nzadot}\sum_{i=1}^NW_{i,A}(z_{A})\sum_{z_{B}}W_{i,B}(z_{B})Y_i(z_{A},z_{B})\Bigg|\bm{W}_A(z_A)\right]\right)\\
&= E\left[\frac{1}{\Nzadot^2}\sum_{i=1}^NW_{i,A}(z_{A})\pi_B(1-\pi_B)\left(Y_i(z_{A},+1_{B})-Y_i(z_{A},-1_{B})\right)^2\right]\\
& + Var\left(\frac{1}{\Nzadot}\sum_{i=1}^NW_{i,A}(z_{A})\left(\pi_BY_i(z_{A},+1)+(1-\pi_B)Y_i(z_{A},-1)\right)\right)\\
&= \frac{\pi_B(1-\pi_B)}{N\Nzadot}\sum_{i=1}^N\left(Y_i(z_{A},+1)-Y_i(z_{A},-1)\right)^2\\
& + \frac{N-\Nzadot}{N\Nzadot}\sum_{i=1}^N\frac{\left(\pi_BY_i(z_{A},+1)+(1-\pi_B)Y_i(z_{A},-1)-\left[\pi_B\overline{Y}(z_{A},+1)+(1-\pi_B)\overline{Y}(z_{A},-1)\right]\right)^2}{(N-1)}
\end{align*}

To simplify notation, let $X_i(z_A) = \pi_BY_i(z_A,+1_B)+(1-\pi_B)Y_i(z_A,-1_B)$.
\begin{align*}
&Cov\left(\overline{Y}^{\text{obs}}_1(+1_A, \cdot), \overline{Y}^{\text{obs}}_1(-1_A, \cdot)\right)\\
=&E\left[Cov\left(\overline{Y}^{\text{obs}}_1(+1_A, \cdot), \overline{Y}^{\text{obs}}_1(-1_A, \cdot)|\bm{W}_A(+1_A)\right)\right]\\
+&Cov\left(E\left[\overline{Y}^{\text{obs}}_1(+1_A, \cdot)|\bm{W}_A(+1_A)\right], E\left[\overline{Y}^{\text{obs}}_1(-1, \cdot)|\bm{W}_A(+1_A)\right]\right)\\
&=0+Cov\left(\frac{1}{\Nplusdot}\sum_{i=1}^NW_{i,A}(+1)X_i(+1_A), \frac{1}{\Nminusdot}\sum_{i=1}^NW_{i,A}(-1_A)X_i(-1_A)\right)\\
&=\frac{1}{\Nplusdot\Nminusdot}\Big(-\sum_{i=1}^N\frac{\Nplusdot\Nminusdot}{N^2}X_i(+1_A)X_i(-1_A)+\sum_{i=1}^N\sum_{i' \neq i}\frac{\Nplusdot\Nminusdot}{N^2(N-1)}X_i(+1_A)X_{i'}(-1_A)\Big)\\
&=-\frac{1}{N^2}\Big(\sum_{i=1}^NX_i(+1_A)X_i(-1_A)-\sum_{i=1}^N\sum_{i' \neq i}\frac{1}{N-1}X_i(+1_A)X_{i'}(-1_A)\Big)\\
&=-\frac{1}{N(N-1)}\sum_{i=1}^N\left(X_i(+1_A)-\overline{X}(+1_A)\right)\left(X_i(-1_A)-\overline{X}(-1_A)\right)
\end{align*}
In the above simplification we used the fact that $Cov\left(\overline{Y}_1(+1_A, \cdot), \overline{Y}_1(-1_A, \cdot)|\bm{W}_A(+1_A)\right)=0$ because all of the $W_{i,B}$ for different units are independent and for the same unit $W_i(+1_A)W_i(-1_A)=0$.

This gives us
\begin{align*}
Var\Big(\widehat{\theta}_{A,1}\Big)&= \sum_{z_A}\frac{\pi_B(1-\pi_B)}{N\Nzadot}\sum_{i=1}^N\left(\theta_{i,B|z_A}\right)^2\\
& + \sum_{z_A}\frac{N-\Nzadot}{N\Nzadot}\sum_{i=1}^N\frac{\left(X_i(z_{A})-\overline{X}(z_{A})\right)^2}{(N-1)}\\
&+\frac{2}{N(N-1)}\sum_{i=1}^N\left(X_i(+1_{A})-\overline{X}(+1_{A})\right)\left(X_i(-1_{A})-\overline{X}(-1_{A})\right)\end{align*}

Next we reformat this expression of variance to make it easier to work with.
We break the variance into 3 parts as follows:
\begin{align*}
Var\Big(\widehat{\theta}_{A,1}\Big)&= \underbrace{\sum_{z_A}\frac{\pi_B(1-\pi_B)}{N\Nzadot}\sum_{i=1}^N\left(\theta_{i,B|z_A}\right)^2}_{\textbf{A}}\\
& + \underbrace{\sum_{z_A}\frac{N-\Nzadot}{N\Nzadot}\sum_{i=1}^N\frac{\left(X_i(z_{A})-\overline{X}(z_{A})\right)^2}{(N-1)}}_{\textbf{B1}}\\
&+\underbrace{\frac{2}{N(N-1)}\sum_{i=1}^N\left(X_i(+1_{A})-\overline{X}(+1_{A})\right)\left(X_i(-1_{A})-\overline{X}(-1_{A})\right)}_{\textbf{B2}}.
\end{align*}

First denote
\begin{align*}
S_X^2 &= \frac{1}{N-1}\sum_{i=1}^N\left(X_i(+1_{A}) - X_i(-1_{A}) -\left(\overline{X}(+1_{A})-\overline{X}(-1_{A})\right)\right)^2\\
&=\frac{1}{N-1}\left[\sum_{z_A}\sum_{i=1}^N\left(X_i(z_{A})-\overline{X}(z_{A})\right)^2-2\sum_{i=1}^N\left(X_i(+1_{A})-\overline{X}(+1_{A})\right)\left(X_i(-1_{A})-\overline{X}(-1_{A})\right)\right].
\end{align*}

Note that terms \textbf{B1} and \textbf{B2} together are the Neyman variance we would obtain if we ran a single factor experiment using $X_i(+1_{A})$ and $X_i(-1_{A})$ as potential outcomes.
So we can rewrite this variance expression using $S_X^2$ in the usual way \citep[see][Chapter 6]{Imbens2015}.
Hence for the combination of terms \textbf{B1} and \textbf{B2} we have
\begin{align*}
&\textbf{B1} + \textbf{B2}\\
&=\sum_{z_A}\frac{N-\Nzadot}{N\Nzadot}\sum_{i=1}^N\frac{\left(X_i(z_{A})-\overline{X}(z_{A})\right)^2}{(N-1)}+\frac{2}{N(N-1)}\sum_{i=1}^N\left(X_i(+1_{A})-\overline{X}(+1_{A})\right)\left(X_i(-1_{A})-\overline{X}(-1_{A})\right)\\
&=\sum_{z_{A}}\frac{1}{\Nzadot}\frac{1}{N-1}\sum_{i=1}^N\left(X_i(z_{A})-\overline{X}(z_{A})\right)^2 - \frac{S_X^2}{N}.
\end{align*}

Now we expand $\frac{1}{N-1}\sum_{i=1}^N\left(X_i(z_A) - \overline{X}(z_A)\right)^2$:
\begin{align*}
&\frac{1}{N-1}\sum_{i=1}^N\left(X_i(z_A) - \overline{X}(z_A)\right)^2\\
 &= \frac{1}{N-1}\sum_{i=1}^N\left(\sum_{z_B}E[W_{i,B}(z_B)]\left(Y_i(z_A, z_B) - \overline{Y}(z_A, z_B)\right)\right)^2\\
&= \frac{1}{N-1}\sum_{i=1}^N\sum_{z_B}E[W_{i,B}(z_B)]^2\left(\left(Y_i(z_A, z_B) - \overline{Y}(z_A, z_B)\right)\right)^2\\
& +  \frac{2\pi_B(1-\pi_B)}{N-1}\sum_{i=1}^N\left(Y_i(z_A, +1_B) - \overline{Y}(z_A, +1_B)\right)\left(Y_i(z_A, -1_B) - \overline{Y}(z_A, -1_B)\right)\\
&= \sum_{z_B}E[W_{i,B}(z_B)]^2S^2(z_A, z_B)+  2\pi_B(1-\pi_B)S^2\left((z_A, +1_B), (z_A, -1_B)\right).
\end{align*}

We similarly expand the term $\frac{1}{N-1}\sum_{i=1}^N\left(X_i(+1_{A})-\overline{X}(+1_{A})\right)\left(X_i(-1_{A})-\overline{X}(-1_{A})\right)$ as follows:
\begin{align*}
&\frac{1}{N-1}\sum_{i=1}^N\left(X_i(+1_{A})-\overline{X}(+1_{A})\right)\left(X_i(-1_{A})-\overline{X}(-1_{A})\right)\\
&=\frac{1}{N-1}\sum_{i=1}^N\sum_{z_{B1}}\sum_{z_{B2}}E[W_{i,B}(z_{B1})]E[W_{i,B}(z_{B2})]\left(Y_i(+1_A, z_{B1}) - \overline{Y}(+1_A, z_{B1})\right)\left(Y_i(-1_A, z_{B2}) - \overline{Y}(-1_A, z_{B2})\right)\\
&=\sum_{z_{B1}}\sum_{z_{B2}}E[W_{i,B}(z_{B1})]E[W_{i,B}(z_{B2})]S^2\left((+1_A, z_{B1}),(-1_A, z_{B2})\right).
\end{align*}

Then our expression of $S_X^2$ is
\begin{align*}
S_X^2&= \frac{1}{N-1}\sum_{i=1}^N\left(X_i(+1_{A}) - X_i(-1_{A}) -\left(\overline{X}(+1_{A})-\overline{X}(-1_{A})\right)\right)^2\\
&= \frac{1}{N-1}\sum_{i=1}^N\left(\sum_{z_B}E[W_{i,B}(z_B)]\left(Y_i(+1_A, z_B) - Y_i(-1_A, z_B) - \left(\overline{Y}(+1_A, z_B)-\overline{Y}(-1_A, z_B)\right)\right)\right)^2\\
&= \frac{1}{N-1}\sum_{i=1}^N\left(\sum_{z_B}E[W_{i,B}(z_B)]\left(\theta_{i,A|z_B} -\theta_{A|z_B}\right)\right)^2.
\end{align*}

Putting this together,
\begin{align*}
Var\Big(\widehat{\theta}_{A,1}\Big)&= \sum_{z_A}\frac{\pi_B(1-\pi_B)}{N\Nzadot}\sum_{i=1}^N\left(\theta_{i,B|z_A}\right)^2\\
&+\sum_{z_{A}}\frac{1}{\Nzadot}\left[\sum_{z_B}E[W_{i,B}(z_B)]^2S^2(z_A, z_B)+  2\pi_B(1-\pi_B)S^2\left((z_A, +1_B), (z_A, -1_B)\right)\right]\\
& - \frac{1}{N-1}\sum_{i=1}^N\left(\sum_{z_B}E[W_{i,B}(z_B)]\left(\theta_{i,A|z_B} -\theta_{A|z_B}\right)\right)^2.
\end{align*}
\end{proof}

\subsection*{Proof of Theorem \ref{thm:caseIvarbias}}\label{supp:bias_var_case_i}

\begin{proof}
Our variance estimator uses
\begin{align*}
s^2(z_{A}) &= \frac{1}{N-1}\sum_{i:W_{i,A}(z_A)=1}\left(Y_i^{obs} - \overline{Y}_1^{obs}(z_{A}, \cdot)\right)^2\\
&= \frac{1}{N-1}\sum_{i:W_{i,A}(z_A)=1}\left(\sum_{z_B}W_{i,B}(z_B)Y_i(z_A, z_B) - \overline{Y}_1^{obs}(z_{A}, \cdot)\right)^2.
\end{align*}

We can condition on $\bm{W}_B(z_B)$, the vector of indicators of which units are assigned to which level of factor B.
This allows us to use the fact that $s^2(z_A)$ is then an unbiased estimator for the variance of potential outcomes under treatment $z_A$ where we fix the treatment level of factor $B $ for each unit.
That is, we can write the following:
\begin{align*}
E[s^2(z_A)|\bm{W}_B(z_B)] &= \frac{1}{N-1}\sum_{i=1}^N\left(\sum_{z_B}W_{i,B}(z_B)Y_i(z_A, z_B) - \frac{\sum_{i=1}^N\sum_{z_B}W_{i,B}(z_B)Y_i(z_A, z_B)}{N}\right)^2\\
&= \frac{1}{N-1}\sum_{i=1}^N\left(\sum_{z_B}\left[W_{i,B}(z_B)Y_i(z_A, z_B) - \frac{\sum_{i=1}^NW_{i,B}(z_B)Y_i(z_A, z_B)}{N}\right]\right)^2.
\end{align*}
Denote
\[\overline{Y}^{obs*}_{B}(z_A,z_B) =\frac{\sum_{i=1}^NW_{i,B}(z_B)Y_i(z_A, z_B)}{N}.\]
Note that we do not actually observe this quantity, the \textit{obs} just indicates that it still involves a random assignment mechanism.

Then we can write
\begin{align*}
&E[s^2(z_A)|\bm{W}(z_B)] \\
&= \frac{1}{N-1}\sum_{i=1}^N\left(\sum_{z_B}\left[W_{i,B}(z_B)Y_i(z_A, z_B) - \overline{Y}^{obs*}_{B}(z_A,z_B)\right]\right)^2\\
&= \sum_{z_B}\underbrace{\frac{1}{N-1}\sum_{i=1}^N\left(W_{i,B}(z_B)Y_i(z_A, z_B) - \overline{Y}^{obs*}_{B}(z_A,z_B)\right)^2}_{\textbf{A}}\\
&+2\underbrace{\frac{1}{N-1}\sum_{i=1}^N\left(W_{i,B}(+1_B)Y_i(z_A, +1_B) -\overline{Y}^{obs*}_{B}(z_A,+1_B)\right)\left(W_{i,B}(-1_B)Y_i(z_A, -1_B) - \overline{Y}^{obs*}_{B}(z_A,-1_B)\right)}_{\textbf{B}}.
\end{align*}
Let's start by simplifying term \textbf{A}.
\begin{align*}
&\frac{1}{N-1}\sum_{i=1}^N\left(W_{i,B}(z_B)Y_i(z_A, z_B) - \overline{Y}^{obs*}_{B}(z_A,z_B)\right)^2\\
&=\frac{1}{N-1}\left[\sum_{i=1}^NW_{i,B}(z_B)Y_i(z_A, z_B)^2 - N\overline{Y}^{obs*}_{B}(z_A,z_B)^2\right]\\
\end{align*}

Eventually we want to find the expectation of $s^2(z_A)$ over assignment to factor $B$. We start by finding the expectation of $\overline{Y}^{obs*}_{B}(z_A,z_B)^2$.
\begin{align*}
&E\left[\overline{Y}^{obs*}_{B}(z_A,z_B)^2\right]\\
&=E\left[\left(\frac{\sum_{i=1}^NW_{i,B}(z_B)Y_i(z_A, z_B)}{N}\right)^2\right]\\
&=E\left[\frac{\sum_{i=1}^NW_{i,B}(z_B)Y_i(z_A, z_B)^2}{N^2}+\frac{\sum_{i=1}^N\sum_{j \neq i}W_{i,B}(z_B)W_{j,B}(z_B)Y_i(z_A, z_B)Y_j(z_A, z_B)}{N^2}\right]
\end{align*}

Taking the expectation over factor $B$ for the whole term \textbf{A},
\begin{align*}
&E\left[\frac{1}{N-1}\sum_{i=1}^N\left(W_{i,B}(z_B)Y_i(z_A, z_B) - \overline{Y}^{obs*}_{B}(z_A,z_B)\right)^2\right]\\
&=\frac{1}{N-1}\left[\frac{N-1}{N}\sum_{i=1}^NE[W_{i,B}(z_B)]Y_i(z_A, z_B)^2 - \frac{1}{N}\sum_{i=1}^N\sum_{j \neq i}E[W_{i,B}(z_B)]^2Y_i(z_A, z_B)Y_j(z_A, z_B)\right]\\
&=\frac{1}{N-1}\Big[\frac{N-1}{N}\sum_{i=1}^NE[W_{i,B}(z_B)]^2Y_i(z_A, z_B)^2 - \frac{1}{N}\sum_{i=1}^N\sum_{j \neq i}E[W_{i,B}(z_B)]^2Y_i(z_A, z_B)Y_j(z_A, z_B)\\
& + \frac{N-1}{N}\sum_{i=1}^N\pi_B(1-\pi_B)Y_i(z_A, z_B)^2 \Big]\\
&=E[W_{i,B}(z_B)]^2S^2(z_A, z_B) + \frac{\pi_B(1-\pi_B)}{N}\sum_{i=1}^NY_i(z_A, z_B)^2.
\end{align*}

Now we turn to term \textbf{B}.
\begin{align*}
&\frac{1}{N-1}\sum_{i=1}^N\left(W_{i,B}(+1_B)Y_i(z_A, +1_B) -\overline{Y}^{obs*}_{B}(z_A,+1_B)\right)\left(W_{i,B}(-1_B)Y_i(z_A, -1_B) - \overline{Y}^{obs*}_{B}(z_A,-1_B)\right)\\
&=\frac{1}{N-1}\sum_{i=1}^N\Big(\overline{Y}^{obs*}_{B}(z_A,+1_B)\overline{Y}^{obs*}_{B}(z_A,-1_B) - W_{i,B}(+1_B)Y_i(z_A, +1_B)\overline{Y}^{obs*}_{B}(z_A,-1_B)\\
&\quad -W_{i,B}(-1_B)Y_i(z_A, -1_B)\overline{Y}^{obs*}_{B}(z_A,+1_B)\Big)\\
&=\frac{1}{N-1}\left(N\overline{Y}^{obs*}_{B}(z_A,+1_B)\overline{Y}^{obs*}_{B}(z_A,-1_B) - 2N\overline{Y}^{obs*}_{B}(z_A,+1_B)\overline{Y}^{obs*}_{B}(z_A,-1_B)\right)\\
&=-\frac{N}{N-1}\overline{Y}^{obs*}_{B}(z_A,+1_B)\overline{Y}^{obs*}_{B}(z_A,-1_B)
\end{align*}

Next we take the expectation of this result over factor $B$.
\begin{align*}
&E\left[-\frac{N}{N-1}\overline{Y}^{obs*}_{B}(z_A,+1_B)\overline{Y}^{obs*}_{B}(z_A,-1_B)\right]\\
 &=-\frac{N}{N-1}E\left[\frac{1}{N^2}\left(\sum_{i=1}^NW_{i,B}(+1_B)Y_i(z_A, +1_B)\right)\left(\sum_{i=1}^NW_{i,B}(-1_B)Y_i(z_A, -1_B)\right)\right]\\
    &=-\frac{N}{N-1}E\left[\frac{1}{N^2}\sum_{i=1}^N\sum_{j\neq i}W_{i,B}(+1_B)W_{j,B}(-1_B)Y_i(z_A, +1_B)Y_j(z_A, -1_B)\right]\\
        &=-\frac{\pi_B(1-\pi_B)}{N(N-1)}\sum_{i=1}^N\sum_{j\neq i}Y_i(z_A, +1_B)Y_j(z_A, -1_B)
\end{align*}

Putting it all together, we have
\begin{align*}
E[s^2(z_A)]&=\sum_{z_B}E[W_{i,B}(z_B)]^2S^2(z_A, z_B) + \sum_{z_B}\frac{\pi_B(1-\pi_B)}{N}\sum_{i=1}^NY_i(z_A, z_B)^2\\
& -\frac{2\pi_B(1-\pi_B)}{N(N-1)}\sum_{i=1}^N\sum_{j\neq i}Y_i(z_A, +1_B)Y_j(z_A, -1_B).
\end{align*}

Hence
\begin{align*}
E\left[\widehat{Var}(\widehat{\theta}_{A,1})\right]&=\sum_{z_A}\frac{E[s^2(z_A)]}{N_{z_A \cdot}}\\
&=\sum_{z_A}\frac{1}{N_{z_A \cdot}}\Bigg[ \sum_{z_B}E[W_{i,B}(z_B)]^2S^2(z_A, z_B) + \sum_{z_B}\frac{\pi_B(1-\pi_B)}{N}\sum_{i=1}^NY_i(z_A, z_B)^2\\
& -\frac{2\pi_B(1-\pi_B)}{N(N-1)}\sum_{i=1}^N\sum_{j\neq i}Y_i(z_A, +1_B)Y_j(z_A, -1_B)\Bigg]\\
&=\sum_{z_A}\frac{1}{N_{z_A \cdot}}\Bigg[ \sum_{z_B}E[W_{i,B}(z_B)]^2S^2(z_A, z_B) + 2\pi_B(1-\pi_B)S^2\left((z_A, -1_B),(z_A, +1_B)\right)\\
& \quad \quad+ \frac{\pi_B(1-\pi_B)}{N}\sum_{i=1}^N\theta_{i,B|z_A}^2\Bigg].
\end{align*}

Recalling the form of the true variance, it is easy to see that
\begin{align*}
E\left[\widehat{Var}(\widehat{\theta}_{A,1})\right] - Var(\widehat{\theta}_{A,1}) &=\frac{1}{N(N-1)}\sum_{i=1}^N\left(\sum_{z_B}E[W_{i,B}(z_B)]\left(\theta_{i,A|z_B}-\theta_{A|z_B}\right)\right)^2.
\end{align*}
This expression is exactly the bias of the standard Neyman variance estimator if we ran an experiment with $X_i(z_A)$ as potential outcomes.
\end{proof}

\subsection{Alternative formulation of variance}

Define
\begin{align*}
\nonumber \mathds{S}^2(\bz, \bz^*) = -\Bigg[&\sum_{i=1}^{N} Cov\left(W_i(\bz), W_{i}(\bz^*)\right)Y_{i}(\bz) Y_{i}(\bz^*)
+ \sum_{i=1}^{N}\sum_{i' \neq i}Cov\left(W_{i}(\bz), W_{i'}(\bz^*)\right)  Y_i(\bz) Y_{i'}(\bz^*)\Bigg]
\end{align*}
and
\begin{align*}
\mathds{S}^2(\bz) &= \sum_{i=1}^{N} Var\left(W_i(\bz)\right)Y_{i}(\bz)^2 + \sum_{i=1}^{N}\sum_{i' \neq i}Cov\left(W_{i}(\bz), W_{i'}(\bz)\right)  Y_i(\bz) Y_{i'}(\bz).
\end{align*}

\begin{lemma}[Expectation and variance of $\Nbz\overline{Y}^{\text{obs}}(\bz)$]
\label{lem:yobs_exp_var}

\begin{align*}
E\left[\Nbz\overline{Y}^{\text{obs}}(\bz)\right]
=& N_{z_A}E[W_{i,B}(z_B)]\overline{Y}(\bz)\\
Var\left(\Nbz\overline{Y}^{\text{obs}}(\bz)\right) 
= &\mathds{S}^2(\bz)\\
Cov\left(\Nbz\overline{Y}^{\text{obs}}(\bz), \Nbzstar\overline{Y}^{\text{obs}}(\bz^*) \right)
=& -\mathds{S}^2(\bz, \bz^*)
\end{align*}

\end{lemma}

\begin{lemma}[Properties of $\overline{Y}^{\text{obs}}_{1}(z_{A}, \cdot)$]
\label{lem:yobs1_prop}
\begin{align*}
E\Big[\overline{Y}^{\text{obs}}_{1}(z_{A}, \cdot)\Big] 
&= \pi_B\overline{Y}(z_{A}, +1_{B}) + (1- \pi_B)\overline{Y}(z_{A}, -1_{B})\\
Var\left(\overline{Y}^{\text{obs}}_{1}(z_{A}, \cdot)\right)
&= \frac{1}{\Nzadot^2}\Bigg[ \sum_{z_B} \mathds{S}^2\left(z_A, z_B\right) - \sum_{z_{B1} \neq z_{B2}} \mathds{S}^2((z_{A}, z_{B1}), (z_{A}, z_{B2}))\Bigg]\\
Cov\left(\overline{Y}^{\text{obs}}_{1}(+1_{A}, \cdot), \overline{Y}^{\text{obs}}_{1}(-1_{A}, \cdot)\right) &= -\frac{1}{2\Nplusdot\Nminusdot}\sum_{\bz, \bz^* \in \mathcal{Z_A^{\neq}}}\mathds{S}^2(\bz, \bz^*).
\end{align*}
\end{lemma}

\begin{proof}
Expectation:
\begin{align*}
E\Big[\overline{Y}^{\text{obs}}_1(z_A, \cdot)\Big] 
&= \frac{1}{\Nzadot}\sum_{z_B}\sum_{i=1}^NE\Big[W_{i, A}(z_A)\Big]E\Big[W_{i, B}(z_B)\Big]Y_i(z_A, z_B)\\
&= \frac{1}{N}\sum_{i=1}^N\left(\pi_BY_i(z_A, +1_B) + (1-\pi_B)Y_i(z_A, -1_B)\right)
\end{align*}
Variance:
\begin{align*}
Var\left(\overline{Y}^{\text{obs}}_1(z_A, \cdot)\right) =& Var\left(\frac{1}{\Nzadot}\sum_{z_{B}}\sum_{i=1}^{N}Y_{i}^{\text{obs}}(z_{A}, z_{B})\right)\\
=&\frac{1}{\Nzadot^2}\Bigg[Var\left(N_{z_A, +}\overline{Y}^{\text{obs}}(z_{A}, +1) \right) + Var\left(N_{z_A, -}\overline{Y}^{\text{obs}}(z_A, -1_B)\right)\\ 
&+ 2Cov\left(N_{z_A, +}\overline{Y}^{\text{obs}}(z_{A}, +1),N_{z_A, -}\overline{Y}^{\text{obs}}(z_A, -1_B)\right)\Bigg]\\
=&\frac{1}{\Nzadot^2}\left[\sum_{z_B} \mathds{S}^2(z_{A}, z_{B}) - \sum_{z_{B1} \neq z_{B2}}\mathds{S}^2((z_{A}, z_{B1}), (z_{A}, z_{B2}))\right]
\end{align*}
Covariance:
\begin{align*}
 Cov\left(\overline{Y}^{\text{obs}}_{1}(+1_A, \cdot), \overline{Y}^{\text{obs}}_{1}(-1_A, \cdot)\right)
= -\frac{1}{2\Nplusdot\Nminusdot}\sum_{\bz, \bz^* \in \mathcal{Z_A^{\neq}}}\mathds{S}^2(\bz, \bz^*)
\end{align*}

\end{proof}

Finally, define the following sets:
\begin{align*}
\mathcal{Z_A^{=}} &= \{ (z_{A1}, z_{B1}), (z_{A2}, z_{B2}): z_{A1} = z_{A2}, z_{B1},z_{B2} \in \{-1_{B},+1_{B}\}\}\\ 
\mathcal{Z_A^{\neq}} &= \{ (z_{A1}, z_{B1}), (z_{A2}, z_{B2}): z_{A1} \neq z_{A2}, z_{B1},z_{B2} \in \{-1_{B},+1_{B}\}\}\\ 
\mathcal{Z_A^+} &= \{(z_{A1}, z_{B1}), (z_{A2}, z_{B2}) \in \mathcal{Z_A^{=}}: z_{A1} = z_{A2} = +1_{A}\}\\
\mathcal{Z_A^-} &= \{(z_{A1}, z_{B1}), (z_{A2}, z_{B2}) \in \mathcal{Z_A^{=}}: z_{A1} = z_{A2} = -1_{A}\}.
\end{align*}

Putting this together,
\begin{align*}
Var\Big(\widehat{\theta}_{A,1}\Big) =& Var\Big(\overline{Y}^{\text{obs}}_{1}(+1_A, \cdot) - \overline{Y}^{\text{obs}}_{1}(-1_A, \cdot)\Big)\\
=& Var\Big(\overline{Y}^{\text{obs}}_{1}(+1_A, \cdot)\Big) + Var\Big(\overline{Y}^{\text{obs}}_{1}(-1_A, \cdot)\Big) -2Cov\Big(\overline{Y}^{\text{obs}}_{1}(+1_A, \cdot), \overline{Y}^{\text{obs}}_{1}(-1_A, \cdot)\Big)\\
=& \frac{1}{\Nplusdot^2}\sum_{z_B} \mathds{S}^2(+1_A, z_B) - \frac{1}{\Nplusdot^2}\sum_{\bz, \bz^* \in \mathcal{Z_A^+} \cap \mathcal{Z_B^{\neq}}}\mathds{S}^2(\bz, \bz^*) + \frac{1}{\Nminusdot^2}\sum_{z_B} \mathds{S}^2(-1,z_B)\\
& - \frac{1}{\Nminusdot^2}\sum_{\bz, \bz^* \in \mathcal{Z_A^-} \cap \mathcal{Z_B^{\neq}}}\mathds{S}^2(\bz, \bz^*)
+ \frac{1}{\Nplusdot\Nminusdot}\sum_{\bz, \bz^* \in \mathcal{Z_A^{\neq}}}\mathds{S}^2(\bz, \bz^*)\\
=& \sum_{\bz} \frac{\mathds{S}^2(\bz)}{\Nzadot^2}
- \sum_{\bz, \bz^{*} \in \mathcal{Z_A^{=}} \cap \mathcal{Z_B^{\neq}}} \frac{\mathds{S}^2(\bz, \bz^*)}{\Nzadot^2} + \sum_{\bz, \bz^{*} \in \mathcal{Z_A^{\neq}}} \frac{\mathds{S}^2(\bz, \bz^*)}{\Nplusdot\Nminusdot}\\
\end{align*}

\section{Proof of results in Section~\ref{sec:caseII}: Case I estimators}\label{supp:caseII}

\subsection*{Properties of components of $\widehat{\theta}_{A,2}$}
\label{supp:am}
For convenience, define 
\[\overline{Y}^{\text{obs}}_{2}(z_{A}, \cdot) = \frac{1}{2}\left(\overline{Y}^{\text{obs}}(z_A, +1_B) + \overline{Y}^{\text{obs}}(z_A, -1_B)\right).\]
\begin{lemma}[Properties of $\overline{Y}^{\text{obs}}_{2}(z_{A}, \cdot)$]
\label{lem:yobs2_prop}

The expectation of $\overline{Y}^{\text{obs}}_{2}(z_{A}, \cdot)$ is
\begin{align*}
E\Big[\overline{Y}^{\text{obs}}_{2}(z_{A}, \cdot)\Big|\Nzaplus >0, \Nzaminus >0\Big] &= \frac{\overline{Y}(z_{A}, +1_{B}) + \overline{Y}(z_{A}, -1_{B})}{2}= \overline{Y}(z_A, \cdot).
\end{align*}

The variance of $\overline{Y}^{\text{obs}}_2(z_A, \cdot)$ is
\begin{align*}
\nonumber & Var\left(\overline{Y}^{\text{obs}}_2(z_A, \cdot)\Big|\Nzaplus>0,\Nzaminus>0\right)\\
&= \frac{1}{4N}\Bigg[ \sum_{z_B} E\left[\frac{N - \Nzazb}{\Nzazb}\Big|\Nzazb > 0\right]S^2(z_A, z_B) -  \sum_{z_{B1} \neq z_{B2}}S^2((z_A, z_{B1}),(z_A, z_{B2}))\Bigg].
\end{align*}

The covariance of $\overline{Y}^{\text{obs}}_{2}(+1_A, \cdot), \overline{Y}^{\text{obs}}_{2}(-1_A, \cdot)$ is
\begin{align*}
Cov\left(\overline{Y}_2^{\text{obs}}(+1_{A}, \cdot), \overline{Y}_2^{\text{obs}}(-1_{A}, \cdot)\Big|N_{\bz}>0 \text{ }\forall \bz \right)=& -\frac{1}{8N}\sum_{\bz, \bz^* \in \mathcal{Z_A^{\neq}}}S^2(\bz,\bz^*).
\end{align*}

\end{lemma}

\begin{proof}
First note that 
\[E[\overline{Y}^{\text{obs}}(\bz)|N_{\bz}>0] = E[E[\overline{Y}^{\text{obs}}(\bz)|N_{\bz}, N_{\bz}>0]|N_{\bz}>0] = \overline{Y}(\bz)\] 
Given this result, it follows that $\overline{Y}^{\text{obs}}_{2}(z_{A}, \cdot)$ is also unbiased.

Now for variance, we have
\begin{align*}
Var&\left(\overline{Y}^{\text{obs}}_2(z_A, \cdot)\Big|\Nzaplus > 0, \Nzaminus > 0\right)\\
=& Var\left(\frac{\overline{Y}^{\text{obs}}(z_A, +1_B) + \overline{Y}^{\text{obs}}(z_A, -1_B)}{2}\Big|\Nzaplus > 0, \Nzaminus > 0\right)\\
=& \frac{1}{4}\Bigg[\underbrace{Var\left(\overline{Y}^{\text{obs}}(z_A, +1_B)\Big|\Nzaplus > 0\right)}_{\textbf{A}_1} + \underbrace{Var\left(\overline{Y}^{\text{obs}}(z_A, -1_B)\Big|\Nzaminus > 0\right)}_{\textbf{A}_2}\\
& + 2\underbrace{Cov\left(\overline{Y}^{\text{obs}}(z_A, +1_B), \overline{Y}^{\text{obs}}(z_A, -1_B)\Big|\Nzaplus > 0, \Nzaminus > 0\right)}_{\textbf{B}}\Bigg].
\end{align*}

Starting with a generalized form of $\textbf{A}_1$ and $\textbf{A}_2$:
\begin{align*}
\textbf{A} 
&=Var\left(\overline{Y}^{\text{obs}}(\bz)\Big|\Nbz > 0\right)\\
&= E\left[Var\left(\overline{Y}^{\text{obs}}(\bz)\Big|\Nbz, \Nbz > 0\right)\Big|\Nbz > 0\right] + Var\left(E\left[\overline{Y}^{\text{obs}}(\bz)\Big|\Nbz, \Nbz > 0\right]\Big|\Nbz > 0\right)\\
&= E\left[Var\left(\overline{Y}^{\text{obs}}(\bz)\Big|\Nbz, \Nbz > 0\right)\Big|\Nbz > 0\right] + Var\left(\overline{Y}(\bz)|N_{\bz} >0\right).
\end{align*}
We note that $Var\left(\overline{Y}(\bz)\right) = 0$, and that $W_i(\bz)$ given $N_{\bz}$ is equivalent to an indicator from a complete randomization, so we can use standard variance results.
\begin{align*}
Var\left(\overline{Y}^{\text{obs}}(\bz)\Big|\Nbz > 0\right) & = E\left[Var\left(\overline{Y}^{\text{obs}}(\bz)\Big|\Nbz, \Nbz > 0\right)\Big|\Nbz > 0\right]\\
&= E\left[\frac{N-\Nbz}{N\Nbz}S^2(\bz)\Big|\Nbz > 0\right]\\
&= E\left[\frac{N-\Nbz}{N\Nbz}\Big|\Nbz > 0\right]S^2(\bz)
\end{align*}

Next, consider a generalized version of $\textbf{B}$:
\begin{align*}
\textbf{B} &=
Cov\left(\overline{Y}^{\text{obs}}(\bz), \overline{Y}^{\text{obs}}(\bz^\star)\Big|\Nbz > 0, \Nbzstar > 0\right)\\
&= Cov\Big(\frac{1}{\Nbz}\sum_{i=1}^{N}Y_{i}(\bz)W_i(\bz),\frac{1}{N_{\bz^\star}}\sum_{i=1}^{N}Y_{i}(\bz^\star)W_i(\bz^\star)\Big|\Nbz > 0, \Nbzstar > 0\Big)\\
&= \sum_{i=1}^{N}Y_{i}(\bz)Y_{i}(\bz^\star) \underbrace{Cov\left(\frac{1}{\Nbz}W_i(\bz), \frac{1}{N_{\bz^\star}}W_i(\bz^\star)\Big|\Nbz > 0, \Nbzstar > 0\right)}_{\textbf{B1}}\\
&+ \sum_{i=1}^{N}\sum_{i' \neq i}Y_{i}(\bz)Y_{i'}(\bz^\star)\underbrace{Cov\left(\frac{1}{\Nbz}W_i(\bz),\frac{1}{\Nbzstar}W_{i'}(\bz^\star)\Big|\Nbz > 0, \Nbzstar > 0\right)}_{\textbf{B2}}
\end{align*}

We find the expression $\textbf{B1}$ for our case, $\bz \neq \bz^\star$.

\begin{align*}
\textbf{B1} &=
Cov\left(\frac{1}{\Nbz}W_i(\bz), \frac{1}{\Nbzstar}W_{i}(\bz^\star)\Big|\Nbz > 0, \Nbzstar > 0\right)\\
&= E\left[\frac{1}{\Nbz\Nbzstar}W_i(\bz)W_{i}(\bz^\star)\Big|\Nbz > 0, \Nbzstar > 0\right] - E\left[\frac{1}{\Nbz}W_i(\bz)\Big|\Nbz > 0\right]E\left[\frac{1}{\Nbzstar}W_{i}(\bz^\star)\Big|\Nbzstar > 0\right]\\
&= - \frac{1}{N^2}
\end{align*}

Similarly, we find the expression $\textbf{B2}$ for our case, $\bz \neq \bz^\star$.
\begin{align*}
\textbf{B2} &=
Cov\left(\frac{1}{\Nbz}W_i(\bz), \frac{1}{\Nbzstar}W_{i'}(\bz^\star)\Big|\Nbz > 0, \Nbzstar > 0\right)\\
&= E\left[\frac{1}{\Nbz\Nbzstar}E\left[W_i(\bz)W_{i'}(\bz^\star)\Big|\Nbz,\Nbzstar, \Nbz > 0, \Nbzstar > 0\right]\right] - \frac{1}{N^2}\\
&= E\left[\frac{1}{\Nbz\Nbzstar}\frac{\Nbz}{N}E\left[W_{i'}(\bz^\star)\Big|\Nbz,\Nbzstar, \Nbz > 0, \Nbzstar > 0, W_i(\bz) = 1\right]\right] - \frac{1}{N^2}\\
&= E\left[\frac{1}{\Nbz\Nbzstar}\frac{\Nbz}{N}\frac{\Nbzstar}{N-1}\Big|\Nbz > 0, \Nbzstar > 0\right] - \frac{1}{N^2}\\
&= \frac{1}{N^2(N-1)}
\end{align*}

Plugging these results in to the generalized version of $\textbf{B}$:
\begin{align*}
\textbf{B} &=
Cov\left(\overline{Y}^{\text{obs}}(\bz), \overline{Y}^{\text{obs}}(\bz^*)\Big| N_{\bz} > 0, N_{\bz^*} > 0\right)\\
&= \sum_{i=1}^{N}\underbrace{\frac{-1}{N^2}}_{\textbf{B1}}Y_{i}(\bz)Y_{i}(\bz^*) + \sum_{i=1}^{N}\sum_{i' \neq i}\underbrace{\frac{1}{N^2(N-1)}}_{\textbf{B2}}Y_{i}(\bz)Y_{i'}(\bz^*)\\
&= -\frac{1}{N}S^2(\bz, \bz^*)
\end{align*}

Finally, we find
\begin{align*}
Var&\left(\overline{Y}^{\text{obs}}_2(z_A, \cdot)\Big| \Nzaplus > 0, \Nzaminus > 0 \right)\\
=& \frac{1}{4}\Bigg[ \underbrace{E\left[\frac{N - \Nzaplus}{\Nzaplus}\Big|\Nzaplus > 0\right]\frac{S^2(z_A, +1_B)}{N}}_{\textbf{A}_1} + \underbrace{E\left[\frac{N - \Nzaminus}{\Nzaminus}\Big|\Nzaminus > 0\right]\frac{S^2(z_A, -1_B)}{N}}_{\textbf{A}_2}\\
&- 2\underbrace{\frac{1}{N}S^2((z_A, +1_B),(z_A, -1_B))}_{\textbf{B}}\Bigg]\\
=& \frac{1}{4N}\Bigg[ \sum_{z_B} E\left[\frac{N - \Nzazb}{\Nzazb}\Big|\Nzazb > 0\right]S^2(z_A, z_B) -  \sum_{z_{B1} \neq z_{B2}}S^2((z_A, z_{B1}),(z_A, z_{B2}))\Bigg].
\end{align*}

We also have
\begin{align*}
Cov&\left(\overline{Y}_2^{\text{obs}}(+1_A, \cdot), \overline{Y}_2^{\text{obs}}(-1_A, \cdot)\Big|N_{\bz}>0 \text{ }\forall \bz\right)\\
=& \frac{1}{4}Cov\Big(\overline{Y}^{\text{obs}}(+1_A, +1_B) + \overline{Y}^{\text{obs}}(+1_A, -1_B), \overline{Y}^{\text{obs}}(-1_A, +1_B) + \overline{Y}^{\text{obs}}(-1_A, -1_B)\Big|N_{\bz}>0 \text{ }\forall \bz\Big)\\
=& \frac{1}{8}\sum_{\bz, \bz^* \in \mathcal{Z_A^{\neq}}}Cov\left(\overline{Y}^{\text{obs}}(\bz), \overline{Y}^{\text{obs}}(\bz^*)\Big|N_{\bz}>0 \text{ }\forall \bz\right)\\
=& -\frac{1}{8N}\sum_{\bz, \bz^* \in \mathcal{Z_A^{\neq}}}S^2(\bz,\bz^*).
\end{align*}
\end{proof}

\subsection*{Proof of Theorem~\ref{thm:theta2_hat_var}}

\begin{proof}
First, use results from the previous section to find the following form of the variance:
\begin{align*}
Var&\left(\widehat{\theta}_{A, 2}\Big|N_{\bz}>0 \text{ }\forall \bz\right)\\
=& Var\left(\overline{Y}^{\text{obs}}_2(+1,\cdot) - \overline{Y}^{\text{obs}}_2(-1,\cdot)\Big|N_{\bz}>0 \text{ }\forall \bz\right)\\
=& Var\left(\overline{Y}^{\text{obs}}_2(+1,\cdot)\Big| N_{+,+} >0, N_{+,-} > 0 \right) + Var\left(\overline{Y}^{\text{obs}}_2(-1,\cdot) \Big| N_{-,+} >0, N_{-,-} > 0 \right)\\
&- 2Cov\left(\overline{Y}^{\text{obs}}_2(+1,\cdot), \overline{Y}^{\text{obs}}_2(-1,\cdot)\Big|N_{\bz}>0 \text{ }\forall \bz\right)\\
=&\frac{1}{4N}\left[ \sum_{\bz} E\left[\frac{N - N_\bz}{N_\bz}\Big|N_{\bz}>0\right] S^2(\bz) -  \sum_{\bz\neq \bz^*}(z_{A1}*z_{A2})S^2(\bz,\bz^*)\right].
\end{align*}

Next note that
\begin{align*}
    S^2_A &= \frac{1}{N-1}\sum_{i=1}^N\left(\sum_{\bz: z_A=+1_A}\frac{1}{2}(Y_i(\bz) - \overline{Y}(\bz)) -\sum_{\bz: z_A=-1_A}\frac{1}{2}(Y_i(\bz) - \overline{Y}(\bz))\right)^2\\
    &=\frac{1}{4}\frac{1}{N-1}\left(\sum_{\bz}(Y_i(\bz) - \overline{Y}(\bz))^2 + \sum_{\bz\neq \bz^*}(z_{A1}*z_{A2})(Y_i(\bz) - \overline{Y}(\bz))(Y_i(\bz^\star) - \overline{Y}(\bz^\star)) \right)\\
    &=\frac{1}{4}\left(\sum_{\bz}S^2(\bz) + \sum_{\bz\neq \bz^*}(z_{A1}*z_{A2})S^2(\bz,\bz^*)\right).
\end{align*}

So we have the result that
\begin{align*}
Var\left(\widehat{\theta}_{A, 2}\Big|N_{\bz}>0 \text{ }\forall \bz\right)&=
\sum_{\bz} E\left[\frac{1}{4N_\bz}\Big|N_{\bz}>0\right]S^2(\bz) -\frac{1}{N}S^2_A.
\end{align*}
\end{proof}

\textbf{Proof of Lemma~\ref{lem:mean_of_Nz_inv}}
\begin{proof}
First note that $N_{z_A, z_B} \sim Binomial\left(\Nzadot, E[W_{i,B}(z_B)]\right)$.
So then for $n \in (1,\dots,\Nzadot - 1)$,
\begin{align*}
    P(N_{z_A, z_B} = n|N_{\bz}>0) &= \frac{P(N_{\bz}>0|N_{z_A, z_B} = n)P(N_{z_A, z_B} = n)}{P(N_{\bz}>0)}\\
    &= \frac{P(N_{z_A,z_B}>0, N_{z_A,-z_B}>0|N_{z_A, z_B} = n)P(N_{z_A, z_B} = n)}{P(N_{z_A,z_B}>0, N_{z_A,-z_B}>0)}\\
    &= \frac{P(N_{z_A, z_B} = n)}{1-E[W_{i,B}(z_B)]^{\Nzadot} - (1-E[W_{i,B}(z_B)])^{\Nzadot}}\\
    &= \frac{P(N_{z_A, z_B} = n)}{1-\pi_B^{\Nzadot} - (1-\pi_B)^{\Nzadot}}
\end{align*}

Then we have,
\begin{align*}
    E\left[\frac{1}{N_{z_A, z_B}}|N_{\bz}>0\right] &= \frac{1}{1-\pi_B^{\Nzadot} - (1-\pi_B)^{\Nzadot}}\sum_{n=1}^{\Nzadot-1}\frac{1}{n}{\Nzadot \choose n}E[W_{i,B}(z_B)]^n(1-E[W_{i,B}(z_B)])^{\Nzadot - n}
\end{align*}
\end{proof}

\end{appendices}

\end{document}